\documentclass[oneside,reqno]{amsart}
\usepackage{amssymb}
\usepackage{graphicx}
\usepackage{enumerate}
\newtheorem{theorem}{Theorem}

\newtheorem{lemma}[theorem]{Lemma}

\newenvironment{myproof}[1][\proofname]{\proof[#1]\mbox{}\\*}{\endproof}
\theoremstyle{remark}
\newtheorem*{remark}{Remark}
\newcommand{\be}{\begin{equation}}
\newcommand{\ee}{\end{equation}}
\newcommand{\ea}[1]{\begin{align}#1\end{align}}
\newcommand{\bmt}{\begin{pmatrix}}
\newcommand{\emt}{\end{pmatrix}}
\newcommand{\fd}[1]{\mathbb{#1}}
\newcommand{\mbf}{\mathbf}
\newcommand{\la}{\langle}
\newcommand{\ra}{\rangle}
\DeclareMathOperator{\GL}{GL}
\DeclareMathOperator{\SL}{SL}

\DeclareMathOperator{\PSL}{PSL}
\DeclareMathOperator{\F}{\mathbb{F}}
\DeclareMathOperator{\PGL}{PGL}
\DeclareMathOperator{\tr}{tr}
\DeclareMathOperator{\Tr}{Tr}

\begin{document}

\begin{titlepage}
\begin{center}
\bfseries  Galois Unitaries, Mutually Unbiased Bases, and MUB-balanced states
\end{center}
\vspace{1 cm}
\begin{center} D.M. Appleby
\\
\emph{Centre for Engineered Quantum Systems, School of Physics, The University of Sydney, Sydney, NSW 2006, Australia}
 \end{center}

\vspace{0.2 cm}

\begin{center}
Ingemar Bengtsson\\
 \emph{Stockholms Universitet, AlbaNova,
  Fysikum, S-106 91
  Stockholm, Sweden}
 \end{center}

\vspace{0.2 cm }

\begin{center}
Hoan Bui Dang
\\
\emph{Perimeter Institute for Theoretical Physics and University of Waterloo,  Waterloo, Ontario, Canada}
\end{center}

\vspace{0.2 cm }

\vspace{1 cm}

\begin{center} \textbf{Abstract}

\vspace{0.1cm}

\parbox{12 cm }{A Galois unitary is a generalization of the notion of anti-unitary 
operators. They act only on those vectors in Hilbert space whose entries belong 
to some chosen number field. For Mutually Unbiased Bases the relevant number field 
is a cyclotomic field. By including Galois unitaries we are able to remove a 
mismatch between the finite projective group acting on the bases on the one hand, 
and the set of those permutations of the bases that can be implemented as transformations 
in Hilbert space on the other hand. In particular we show that there exist 
transformations that cycle through all the bases in all dimensions $d = p^n$ where 
$p$ is an odd prime and the exponent $n$ is odd. (For even primes unitary MUB-cyclers 
exist.) These transformations have eigenvectors, which are MUB-balanced states (i.e.\ rotationally symmetric states in the original terminology of Wootters and Sussman) if and 
only if $d = 3$ modulo 4. We conjecture that this construction yields all such states in odd prime power dimension.
}
\vspace{0.35 cm}
\parbox{12 cm }{}
\end{center}
\end{titlepage}

\section{Introduction}
Symmetries in quantum mechanics are described by unitary 
or anti-unitary operators \cite{Wigner}. Recently it was noted that the 
notion of anti-unitary operators can be generalized to something dubbed 
``g-unitaries'', where the ``g'' stands for Galois and means 
that these operations can be applied only to suitably restricted vectors in 
Hilbert space 
\cite{AYAZ}. This new idea deserves to be followed up. Here we will 
show that g-unitaries 
have a role to play when thinking about a problem raised by Arnold in his 
last book \cite{Arnold}. His starting point is the observation that certain finite 
sets can be interpreted as projective lines over finite Galois fields (and 
we add that the name of Galois now occurs for the second time for a separate 
reason). One of Arnold's problems is to understand the intrinsic characterization 
of those permutations of the set that arise as projective permutations. He also 
asks for applications to the natural sciences. We will argue that in fact 
this problem is relevant when one is trying to understand the shape of quantum 
state spaces, and that g-unitaries enter at a key point of the argument.  

What we have in mind is the well known fact that Mutually Unbiased Bases (MUB) 
form a projective line over a finite Galois field, when the dimension $d$ of 
Hilbert space is a power of a prime number $p$. This is connected to the existence of 
the finite affine plane on which discrete Wigner functions are defined \cite{Wootters}.   
There are $d+1$ MUB and altogether $d(d+1)$ unit vectors involved. It turns out that 
the permutations of bases that can be effected by unitary transformations 
are in fact the projective transformations that interested Arnold. However, if 
$d$ is odd there are additional projective transformations that do not arise 
in this way.
It is at this point that g-unitaries enter the 
story to provide some of the missing symmetries.  Moreover they have a geometric 
interpretation in terms of rotations in Bloch space. When viewed 
as projection operators the MUB vectors give rise to a polytope which can be 
described by noting that the bases span $d+1$ totally orthogonal planes 
in Bloch space \cite{Asa}.  It plays a major role in fault tolerant 
quantum computation \cite{Cormick}, and in prime dimensions it is known as the 
stabilizer polytope. To avoid a possible ambiguity in prime power dimensions 
we will instead refer to it as the complementarity polytope. (When the dimension 
is a power of a prime the relevant Heisenberg group yields a set of several 
interlocking MUBs. We focus on just one of them.) The symmetry group of 
the complementarity polytope gives rise to arbitrary permutations of 
the $d+1$ planes spanned by the bases, but there is a natural restriction which 
gives rise to precisely the projective permutations including those effected 
by g-unitaries.   

We will pay special attention to MUB-cyclers, that is transformations of order 
$d+1$ cycling through the entire set of MUB. If $d$ is even unitary MUB 
cyclers exist \cite{Sussman, Seyfarth}. If $d=p^n$ is odd and $d$ = 3 mod 4,
anti-unitary MUB-cyclers exist \cite{Appleby}. 
We will find g-unitary MUB-cyclers for all $d = p^n$ where $n$ is odd. This 
seems to us to be an interesting fact, even though we freely admit that unitary 
MUB-cyclers are the important ones from the point of view of the experimentalist: 
they enable us to reach any MU basis by iterating a single operation in the lab. 
Implementing g-unitaries in the lab will be very hard! However, they remain 
interesting from the point of view of Arnold's problem. 

One important reason why MUB-cyclers deserve attention is that one expects 
vectors invariant under a MUB-cycler to be MUB-balanced states in the terminology of Amburg et al~\cite{ASSW} (or rotationally symmetric states in the original terminology of Wootters and Sussman~\cite{Sussman}). Such states are defined by 
the property that the $d+1$ probability vectors obtained by projection to 
the $d+1$ MUB are identical, up to permutations of the components. MUB-balanced states were originally constructed in even prime power dimension by Wootters and Sussman~\cite{Sussman}, and Amburg et al~\cite{ASSW} recently gave a construction for odd prime power dimensions equal to 3 modulo 4.
Amburg et al took notice of some interesting properties possessed by 
these states, and expressed their surprise 
that such states exist at all. Indeed, according to these authors, MUB-balanced 
states are interesting because they ``have no right to exist''. The same can 
be said for the MUB themselves, and also for the SICs that we will mention below. 

The existence of MUB cycling g-unitaries does not settle 
the existence of MUB-balanced states, since---unlike unitaries---such operators 
need not have any eigenvectors at all. We will establish that the ones we are 
considering do. At this point we fully expected that we would obtain a large 
supply of MUB-balanced states, but further analysis shows that a vector 
invariant under a MUB cycling g-unitary is in fact a MUB-balanced 
state if and only if it is also left invariant by a MUB-cycling anti-unitary. 
The latter exist only in prime power dimensions $d = 3$ mod 4. In these 
dimensions we are able to prove that the MUB-balanced states arising from MUB 
cyclers belong to a single orbit under the Clifford group. The obstruction that 
arises when $d = 1$ mod 4 has to do with the fact that g-unitaries do not 
preserve Hilbert space norm in general \cite{AYAZ}.

The states constructed by Amburg et al.~\cite{ASSW} are representatives of the 
single orbit that we found. We conjecture 
that these authors have in fact found all MUB-balanced states up to the 
action of the Clifford group. This highlights the remarkable 
position of these states within quantum state space. 

MUB-balanced states belong to the wider class of Minimum Uncertainty States 
(MUS) \cite{Sussman}. Geometrically a state is a MUS if the probability vectors 
arising from projections to the MUB have equal length. Simple parameter counting 
suggests that the set of MUS in a Hilbert space of complex dimension $d$ has 
$d-2$ real dimensions. While not very exceptional in themselves, the set of MUS 
does include all the MUB-balanced states, and surprisingly---provided the 
dimension is a prime numer---they also include all Weyl-Heisenberg covariant 
SICs \cite{ADF}. Following these observations Wootters and Sussman and (independently 
of them) one of the present authors (DMA) obtained some results concerning 
MUB-balanced states in odd prime power dimensions.  However these results were 
not published in a journal at the time (though some of them did appear in 
ref. \cite{Sussmanthesis}).  There was then a lapse of six years after which, quite independently, both the present authors and Wootters and Sussman \cite{ASSW} (together with Amburg and Sharma) decided to return to the problem.  The approach taken in ref. 
\cite{ASSW} is very different from ours, and we believe that our approach provides 
a useful complement. 

We have saved a major motivation for this work for last. The whole discussion 
will rest on the fact that, in the natural basis singled out by the group, 
all the entries of the MUB vectors belong to the cyclotomic field generated 
by the roots of unity. The cyclotomic field is an abelian extension of the rational 
field $\fd{Q}$. This may seem like an overly sophisticated way of expressing the 
fact that the entries are roots of unity, but it should be 
remembered that there is every reason to believe that the entries of the vectors 
forming what is known as a SIC-POVM belong to an abelian extension of the real 
quadratic field $\fd{Q}(\sqrt{(d-3)(d+1)})$ \cite{AYAZ}. The latter is in itself 
an extension over the rationals, but its abelian extensions are of a rather mysterious 
kind. In fact they form the subject of Hilbert's 12th problem \cite{Schappacher}. 
The SICs themselves are very distinguished orbits of the Weyl-Heisenberg 
group, and they are prime examples of sets of states that ``have no right to exist''. 
Still, it seems that they do, at least in  low dimensions.  Besides their own intrinsic interest they are  important technically, due to their applications to quantum tomography, quantum cryptography and entanglement detection (see, for example, refs.~\cite{ScottB,Englert,Chen}), and also conceptually, due to their role in the qbist program (see, for example, ref.~\cite{FuchsSchack}).  So their properties and significance deserve to be better
understood. At the moment, their existence in arbitrary dimensions remains a tantalizing 
conjecture only \cite{Zauner, Renes, DMA, Scott}. The notion of g-unitaries arose in 
an attempt to understand them better \cite{AYAZ}, and from this point of view we 
are investigating a toy model for SICs.

\section{Choice of dimensions and organization of our paper}
The story that we have to tell depends sensitively on the dimension $d$ of 
Hilbert space, and becomes especially transparent if $d$ is an odd prime 
number $p$. The case when $d = p^n$ is a power of an odd prime number is partly a 
straightforward generalization but does require rather more in the way of 
notation, and moreover the argument diverges from the $d = p$ case at some key 
points. 

Section 3 introduces the minimal amount of background concerning fields and 
Galois extensions that we will need in this paper. We then introduce the Clifford 
group and its extension to g-unitaries in sections 4 and 5. In order to increase 
readability we confine these two sections to the case when the dimension is an 
odd prime number, and discuss the general case $d = p^n$ (with $p$ odd) in 
section 6. In section 7 we introduce the 
complementarity polytope, and in section 8 we give a geometric interpretation of 
the g-unitaries in terms of its symmetry goup. In section 9 we prove a 
result concerning MUB-cyclers. There is a significant difference depending on 
whether the exponent $n$, where $d = p^n$, is odd or even. In section 10 we 
prove that g-unitary MUB-cyclers do have eigenvectors, and in section 11 we 
establish that these eigenvectors are MUB-balanced states when $d = 3$ mod $4$. 
Since the story we tell becomes involved and (we are afraid) makes some demands 
on the reader's time, we have tried to summarize our main results in words in 
section 12. The results are in fact simple and (we think) appealing.

\section{Fields and field extensions}\label{secfield}
This review section is intended to be a brief introduction to fields, field extensions, Galois automorphisms, and finite fields, for readers with little or no background in field theory. The role of finite fields in the theory of Mutually Unbiased Bases 
was made clear early on \cite{Fields}, but here we will need an infinite number field as well. Since the main purpose is to help our readers quickly grasp the key concepts used in the paper, we avoid unnecessarily technical definitions and derivations as much as we can. Rigorous treatments of the subjects can be found in textbooks on fields 
and Galois theory \cite{Stewart, Milne, Roman}.

\underline{Fields}. A field $\fd{F}$ is a set with two commutative operations, addition and multiplication, that are compatible via distributivity. Furthermore, $\fd{F}$ has an additive identity 0, a multiplicative identity 1, and every element in $\fd{F}$ has an additive inverse and, with the exception of 0, a multiplicative inverse. Thus $\fd{F}$ 
is a group under addition, and $\fd{F}$ with the 0 element removed is a group under 
multiplication. 
Common examples include the field of rational numbers $\fd{Q}$ and the field of real numbers $\fd{R}$, with the usual addition and multiplication. There also exist finite fields, i.e. fields with a finite number of elements. For example, if $p$ is a prime number, then the set of integers $\fd{Z}_p=\{0,1,..,p-1\}$ with addition and multiplication modulo $p$ forms a field, called a prime field.

\underline{Field extensions}. The complex field $\fd{C}$ commonly used in quantum physics is constructed from the real field $\fd{R}$ by adding to it an imaginary number $i$ defined by the property $i^2 = -1$, in other words $i$ is defined to be a root of the real polynomial $x^2+1$. The complex field is then defined as the set of all numbers of the form

\begin{equation} \fd{C} = \{a + i b: a,b \in \fd{R}\}
\end{equation}

\noindent where the sum and product of any two elements $(a_1+ib_1)$ and $(a_2+ib_2)$, using the identity $i^2+1=0$, can be easily worked out to be $(a_1+a_2) + i(b_1+b_2)$ and $(a_1 a_2 - b_1 b_2) + i(a_1 b_2 + a_2 b_1)$, which are clearly also in $\fd{C}$. One can think of a complex number $(a+i b)$ as a 2-component vector $(a,b)$ in a real vector space. The dimensionality of this vector space is equal to the degree of the polynomial defining $i$, namely 2.

More generally, given a field $\fd{F}$ and a number $h \notin \fd{F}$, we can construct a field $\fd{E}$ such that it is the smallest field containing both $h$ and $\fd{F}$, denoted as $\fd{E} \equiv \fd{F}(h)$. $\fd{F}$ is called the ground field, $\fd{E}$ is called the extended field, and the field extension is denoted as $\fd{E}/\fd{F}$ (reads as $\fd{E}$ over $\fd{F}$). We assume that $h$ is algebraic over 
$\fd{F}$, i.e. it is a root of some polynomial with coefficients in $\fd{F}$. Among polynomials that admits $h$ as a root, consider one with the lowest degree, and let $n$ denote its degree. The extended field $\fd{E}$ can then be constructed as

\begin{equation}\label{defextension}
\fd{E} =\{f_0 + f_1 h + ...+f_{n-1} h^{n-1}: f_i \in \fd{F}\} \ .
\end{equation}

\noindent $\fd{E}$ can be thought of as an $n$-dimensional vector space over 
$\fd{F}$, and $n$ is called the degree of the field extension $\fd{E}/\fd{F}$.

\underline{Galois automorphisms}. Given a field extension $\fd{E}/\fd{F}$ ($\fd{E}$ is an extension of $\fd{F}$), a Galois automorphism of $\fd{E}/\fd{F}$ is defined as an automorphism of $\fd{E}$ that fixes the elements in $\fd{F}$. In other words, it is a bijective mapping $g: \fd{E} \to \fd{E}$ that satisfies:

\begin{enumerate}
\item $g(e_1 + e_2) = g(e_1) + g(e_2)$ for all  $e_1,e_2 \in \fd{E}$
\item $g(e_1  e_2) = g(e_1)  g(e_2)$ for all  $e_1,e_2 \in \fd{E}$
\item $g(f) = f$ for all  $f \in \fd{F}$
\end{enumerate}
All such automorphisms form a group called the Galois group of the extension $\fd{E}/\fd{F}$, denoted by Gal$(\fd{E}/\fd{F})$. The order of the Galois group is less than or equal to the degree of the extension, and when they are equal we call such a field extension a Galois extension. All the field extensions considered in this paper are Galois extensions.

Let us go back to the example of the extension of the real field $\fd{R}$ to the complex field $\fd{C}$. If $g: \fd{C} \to \fd{C}$ is a Galois automorphism of the extension $\fd{C}/\fd{R}$, then it satisfies

\begin{equation}
g(i)g(i) = g(i^2) = g(-1) = -1 \ ,
\end{equation}

\noindent which implies either $g(i)=i$ or $g(i)=-i$. Note that the value of $g(i)$ completely specifies the Galois automorphism $g$ because its action on any complex number $(a+ib)$ can be expressed as $g(a+ib) = g(a)+g(i)g(b) = a+g(i)b$. In this case the two Galois automorphisms in the Galois group corresponding to $g(i) = i$ and $g(i) = -i$ are the identity mapping and complex conjugation. This extension is a Galois extension, as the group has order 2, equal to the degree of the extension.

\underline{Cyclotomic fields}. A cyclotomic field is generated by extending the rational field $\fd{Q}$ with an $N$-th root of unity $\omega=e^{2 \pi i/N}$. Although cyclotomic fields can be defined for any $N$, we will restrict ourselves to the case when $N=p$ is a prime number. Then the minimal polynomial of $\omega$ over $\fd{Q}$ is

\begin{equation}
P(x) = 1 + x + ... + x^{p-1} \ .
\end{equation}

\noindent The cyclotomic field $\fd{Q}(\omega)$ is then defined by 

\begin{equation}
\fd{Q}(\omega) = \{q_0 + q_1 \omega + ... + q_{p-2}  \omega^{p-2}: q_i \in \fd{Q}\} \ .
\end{equation}

\noindent The extension $\fd{Q}(\omega)/\fd{Q}$ is of degree $p-1$. Let $g: \fd{Q}(\omega) \to \fd{Q}(\omega)$ be a Galois automorphism of the field extension. Just like in the complex case, $g$ is completely specified by the value of $g(\omega)$. The identity

\begin{equation}
(g(\omega))^p = g(\omega^p) = g(1) = 1 \ ,
\end{equation}

\noindent implies that $g(\omega)= \omega^k$ for some $1\le k \le p-1$. If we specifically denote $g_k$ to be the Galois automorphism that maps $\omega$ into $\omega^k$, then $\{g_k\}_{k=1}^{p-1}$ is the order $p-1$ Galois group of this field extension. One can see that $g_1$ is the identity mapping, and $g_{p-1}$ is complex conjugation.

\underline{Finite fields}. A finite field (somewhat confusingly, also called a Galois field) is a field that has a finite number of elements, called its order. The prime field $\fd{Z}_p$ for any prime number $p$ is an example of a finite field, as previously mentioned. There are also finite fields of other orders. However, it is known that finite fields only exist for which the order is a prime power $p^n$, and for every prime power there exists a unique (up to isomorphism) field of this order. Thus, we can refer to a finite field only by its order, and we shall denote the finite field of order $N$ by $\fd{F}_N$, where $N$ must be a prime power for $\fd{F}_N$ to exist. We will not provide the proof here, but will instead give a concrete example of how to generate $\fd{F}_{p^n}$ by extending the prime field $\fd{F}_p$.

Consider the finite field $\fd{F}_3 = \{0,1,2\}$. If we take the square (modulo $3$) of each element in $\fd{F}_3$ we will get the set $\{0,1\}$, which, if we exclude $0$, is called the set of quadratic residues, i.e. the set of non-zero elements that can be expressed as a square of some element in the field. $2$ is a quadratic non-residue because no element in $\fd{F}_3$ squares to 2, therefore the polynomial $x^2+1 \equiv x^2-2 $ (mod $3$) has no root in $\fd{F}_3$ and is irreducible over $\fd{F}_3$. If we define $\lambda$ to be the root of $x^2+1$, a finite field version of the imaginary number $i$, then we can adjoin $\lambda$ to $\fd{F}_3$ to create the extended field

\begin{equation}
\fd{F}_3(\lambda)= \{a + \lambda b: a,b \in \fd{F}_3\} \ . 
\end{equation}

\noindent One can easily see that $\fd{F}_3(\lambda)\equiv \fd{F}_9$ has 9 elements 
and that its multiplication table can be calculated using the identity $\lambda^2+1=0$. In general, if we start from an irreducible polynomial of degree $r$ in $\fd{F}_p$, we will be able to extend the field to $\fd{F}_{p^r}$.

There are a few basic properties of finite fields relevant to our paper that we would 
like to mention here. First of all every finite field admits one (and 
in fact several) primitive element $\theta$, that is an element such that every non-zero 
element $x$ in the field can be written as $x = \theta^r$ for some choice of the exponent 
$r$. Then, for $\fd{F}_p$ and all extension fields $\fd{F}_{p^n}$, the following facts hold:

\begin{enumerate}
\item $a^{p^n}=a  \hskip 7mm \forall a \in \fd{F}_{p^n}$
\item $(a+b)^p = a^p + b^p  \hskip 7mm \forall a,b \in \fd{F}_{p^n}$
\item $\forall a \in \fd{F}_{p^n}$, $a^p=a$ if and only if $a \in \fd{F}_p$.
\end{enumerate}

\section{The Clifford group and Mutually Unbiased Bases}
This section is again a review of things that are fully described elsewhere 
\cite{DMA, Appleby}. 
We begin in odd prime dimensions $d$, where the Weyl-Heisenberg group has an 
essentially unique representation given 
by a primitive root of unity and by Sylvester's clock and shift matrices 

\begin{equation} \omega = e^{\frac{2\pi i}{d}} \ , \hspace{5mm} 
Z|x\rangle = \omega|x\rangle \ , \hspace{5mm} X|x\rangle = |x+1\rangle \ . 
\end{equation} 

\noindent The labels on the states are integers modulo the dimension $d$, and 
in this and in the following section $d$ is taken to be an odd prime. The group elements 
are best organized into the displacement operators 

\begin{equation} D_{\bf u} = \omega^{\frac{u_1u_2}{2}}X^{u_1}Z^{u_2} \ , \end{equation}

\noindent where ${\bf u}$ is a two component vector with entries $u_1$, $u_2$ chosen 
from the integers modulo $d$. The latter form the finite field $\fd{Z}_d$, and 
$1/2$ denotes the multiplicative inverse of 2 in this field. The reason for inserting 
the phase factors into the definition of 
the displacement operators is that the group law then takes the useful form 

\begin{equation} D_{\bf u}D_{\bf v} = \omega^{\Omega ({\bf u}, {\bf v})} D_{{\bf u} + 
{\bf v}} \ , \hspace{6mm} \Omega ({\bf u}, {\bf v}) = u_2v_1-u_1v_2 \ . 
\label{grouplaw} \end{equation}

\noindent Here $\Omega$ is a symplectic form on the vector space $(\fd{Z}_d)^2$.  

The Clifford group is defined as the normalizer of the Weyl-Heisenberg group within 
the unitary group. Modulo phases it is isomorphic to a semi-direct product of the 
special linear group $\SL(2,\fd{Z}_d)$ (which is isomorphic to the symplectic group) 
acting on the discrete translation group $(Z_d)^2$. 
This action is 

\begin{equation} U_GD_{\bf u}U^{-1}_G = D_{G{\bf u}} \ . \end{equation}

\noindent It is then obvious from the group law 
(\ref{grouplaw}) that the matrices $G$ have to preserve the symplectic form---i.e. 
they must have determinant unity, and belong to $\SL(2, \fd{Z}_d )$. Their unitary representation $U_G$ is fixed up to overall 
phase factors. In this paper we will need the latter too, so we use the metaplectic 
representation which is faithful as opposed to only projective. It is given by 
\cite{Appleby}

\begin{equation}  G = \left( \begin{array}{cc} \alpha & \beta \\ \gamma & \delta 
\end{array} \right) \hspace{3mm} \rightarrow \hspace{3mm} 
\left\{ \begin{array}{lll}
U_G = \frac{e^{i\theta}}{\sqrt{d}}\sum_{r,s} 
\omega^{\frac{1}{2\beta}(\delta r^2 - 2rs + \alpha s^2)} 
|r\rangle \langle s| & \ & \beta \neq 0 \\
\\
U_G = l(\alpha ) \sum_s\omega^{\frac{\alpha \gamma}{2}s^2}|\alpha s\rangle \langle s| & \ & 
\beta = 0 \end{array} \right. 
\label{1} \end{equation}

\noindent where $\alpha \delta - \beta \gamma = 1$ modulo $d$ and 

\begin{equation} e^{i\theta} = - \frac{1}{i^{\frac{d+3}{2}}} l(-\beta) = \left\{ 
\begin{array}{lll} (-1)^kl(-\beta ) & \mbox{if} & d = 4k+1 \\ \\ (-1)^{k+1}il(-\beta ) & 
\mbox{if} & d = 4k+3 \end{array} \right. \label{2} \end{equation}

\begin{equation} l(x) = \left\{ \begin{array}{rll} 1 & \mbox{if} & x \in {\bf Q} \\ 
\\ -1 & \mbox{if} & x \in {\bf N} \ . \end{array} \right. \end{equation}

\noindent Number theorists know $l(x)$ as the Legendre symbol. ${\bf Q}$ is the 
set of quadratic residues, that is to say the set of non-zero integers modulo $d$ 
that can be written as the square of another integer modulo $d$, and ${\bf N}$ is 
the set of quadratic non-residues. These two sets have the same size. 

With the Clifford group in hand (and $d$ still set firmly equal to an odd prime) 
we can obtain a complete set of $d+1$ MUB in Hilbert space. We simply choose 
a vector belonging to the computational basis, and act on it with the Clifford 
group. The resulting orbit consists of $d(d+1)$ unit vectors forming $d+1$ 
MUB. We denote these vectors by 

\begin{equation} |e_r^{(z)}\rangle \ , \hspace{5mm} 
z \in \{ 0,1,\dots , d-1, \infty \} \ , \hspace{5mm} 
r \in \{ 0 , 1, \dots , d-1\}\ . \label{MUvectors} \end{equation}

\noindent where $z$ labels the bases and $r$ labels the vectors in a basis. Thus 

\begin{equation} |\langle e^{(z)}_r|e^{(z')}_{r'}\rangle |^2 = \frac{1}{d} + 
\delta_{z,z'}\left( \delta_{r,r'} - \frac{1}{d}\right) \ . \end{equation}

\noindent This amazing fact is by now sufficiently well known \cite{Ivanovic, 
Fields, Vatan}, so let us just 
mention that each individual basis is an eigenbasis of a cyclic subgroup of 
the Weyl-Heisenberg group. In general the Weyl-Heisenberg group permutes the 
vectors within the bases, while a symplectic unitary transforms the bases among 
themselves. To be precise, given a symplectic matrix of the form given in 
eq. (\ref{1}) the 
bases are permuted by the M\"obius transformation \cite{Subhash}

\begin{equation}z \rightarrow z' = \frac{\alpha z + \beta}{\gamma z + \delta} 
\ . \label{Mobius} \end{equation}

\noindent This can be interpreted as a projective transformation of a finite 
projective line whose points are labelled by $z$, and it explains 
why one of the bases was labelled by $\infty $. In particular 

\be
z' =
\begin{cases}
\infty \qquad & \text{if $z\neq \infty$, $\gamma z + \delta = 0$ or $z=\infty$, $\gamma=0$} 
\\
\frac{\alpha}{\gamma} \qquad & \text{if $z=\infty$, $\gamma \neq 0$ \ .}
\end{cases} 
\ee

\noindent The individual vectors are also permuted by the Clifford group, up to 
phase factors that belong to the cyclotomic field. 

There is an oddity here, since the Clifford group does not give us the most general M\"obius transformation. The latter is of the form (\ref{Mobius}), with 
$\alpha \delta - \beta \gamma \neq 0$ but otherwise unrestricted. Such transformations 
are obtained from the general linear group $\GL(2, \fd{Z}_d)$, consisting of all invertible 
two by two matrices with entries in $\fd{Z}_d$, by taking the quotient with all 
diagonal matrices. They form 
the projective group $\PGL(2, \fd{Z}_d)$. Similarly, the special linear group 
$\SL(2, \fd{Z}_d)$ gives rise to the projective group $\PSL(2, \fd{Z}_d)$, which 
is a proper subgroup of $\PGL$. In fact the respective orders of these groups are 

\begin{equation} |\GL | = (d-1)d(d^2-1) \ , \hspace{5mm} |\PGL| = d(d^2-1) , \end{equation}

\begin{equation} |\SL | = d(d^2-1) \ , \hspace{5mm} |\PSL| = \frac{d(d^2-1)}{2} \ . 
\end{equation}

\noindent Those elements in $\GL$ that do not belong to $\SL$ are easily identified. 
We observe that 

\begin{equation} \left( \begin{array}{cc} 1 & 0 \\ 0 & x^2 \end{array} \right) = 
\left( \begin{array}{cc} x & 0 \\ 0 & x \end{array} \right) 
\left( \begin{array}{cc} x^{-1} & 0 \\ 0 & x \end{array} \right) \ . \end{equation}

\noindent Elements of $\GL$ for which the determinant is a quadratic residue do 
not add anything to $\PGL$ beyond the contribution of $\SL$. But elements whose 
determinant is not a quadratic residue do.

For a prime $d$, it is a number theoretical fact that $-1$ is a quadratic residue modulo $d$ if 
and only if $d = 1$ modulo 4. Matrices $G$ having determinant $-1$ can be represented 
as acting on Hilbert space through anti-unitary transformations \cite{DMA}, and if $d =3$ 
modulo 4 the Clifford group extended to include such elements will yield general 
projective permutations of the bases \cite{Appleby}. But if 
$d = 1$ modulo 4 every other projective permutation is missing. It is at 
this point that g-unitaries enter the story.

\section{The Clifford group extended by ${\rm g}$-unitaries}
When we perform general linear transformations $G$ in the discrete phase 
space symplectic areas will change according to $\Omega ({\bf u}, {\bf v}) 
\rightarrow \Delta \Omega ({\bf u}, {\bf v})$, where $\Delta$ is the determinant of the 
matrix $G$. To stay consistent with the group law (\ref{grouplaw}) such a 
transformation must be accompanied with the Galois automorphism $\omega \rightarrow 
\omega^\Delta$. If not it will not be an automorphism of the Weyl-Heisenberg group, 
as we insist it should be. The question is whether such transformations are 
allowed. 

Consider first the case $\det{G} = \Delta = -1$. In this case we are simply dealing 
with complex conjugation, and more generally with anti-unitary transformations 
of Hilbert space. This is certainly allowed, and takes us to the extended 
Clifford group, which is well understood \cite{DMA}. As we noted in the previous 
section it will give rise to general projective permutations of the bases if 
$d = 3$ modulo 4.

Now consider the case of general $\GL(2, \fd{Z}_p)$ matrices. Clearly 

\begin{equation} \det{G} = \Delta \hspace{5mm} \Rightarrow \hspace{5mm} G = 
\bar{G}K_\Delta \ , \hspace{6mm} \det{\bar{G}} = 1 \ , \hspace{5mm} 
K_\Delta = \left( \begin{array}{cc} 1 & 0 \\ 0 & \Delta 
\end{array} \right) \ . \end{equation}

\noindent So it will be enough to have a representation of the matrix $K_\Delta$. 
We decide that 

\begin{equation} K_x = \left( \begin{array}{cc} 1 & 0 \\ 0 & x 
\end{array} \right) \label{Kx} \end{equation}

\noindent is represented by the automorphism $g_x$,  

\begin{equation} g_x \ : 
\hspace{5mm} \omega \rightarrow \omega^x \ .
\end{equation} 

\noindent This means that 

\begin{equation} U_G|\psi\rangle = U_{\bar{G}}g_\Delta(|\psi\rangle ) \ . \end{equation}

\noindent  The notation $g_x(|\psi\rangle )$ will always mean the vector obtained 
by applying the automorphism $g_x$ to the components of $|\psi\rangle $ in the standard 
basis. The notation $g_x(A)$ is defined similarly  using the matrix elements of the 
operator $A$. This action is defined only on matrices and vectors 
whose entries belong to the cyclotomic field. 

Going back to the explicit expressions in eqs. (\ref{1}-\ref{2}) we see that there 
is a question whether the overall factor $e^{i\theta}/\sqrt{d}$ is in the cyclotomic 
field. To answer it we recall the Gaussian sum 

\begin{equation} \sum_{x = 0}^{p-1}\omega^{x^2} = \left\{ \begin{array}{lll} 
\sqrt{p} & \mbox{if} & p = 4k+1 \\ \\ i\sqrt{p} & \mbox{if} & p = 4k+3 \ . \end{array} 
\right. \end{equation}

\noindent We can rewrite this as  

\begin{equation} \sum_{x \in {\bf Q}}\omega^x - \sum_{x \in {\bf N}}\omega^x = 
\left\{ \begin{array}{lll} \sqrt{p} & \mbox{if} & p = 4k+1 \\ 
\\ i\sqrt{p} & \mbox{if} & p = 4k+3 \ . \end{array} \right. \label{rotp} \end{equation}

\noindent From this we reach the conclusion that the entries in the unitary matrices 
representing the symplectic group do indeed belong to the cyclotomic field, and then 
this is obviously true for the entire Clifford group. This means that we can use 
the Galois automorphisms of the cyclotomic field to represent the group 
$GL(2, \fd{Z}_d)$. 

We must show that the representation is faithful. For this purpose write 
$G = \bar{G}K$, $\det{\bar{G}} = 1$, and consider

\begin{equation} G_1G_2 = \bar{G}_1K_1\bar{G}_2K_2 = 
\bar{G}_1(K_1\bar{G}_2K_1^{-1})K_1K_2 \ , \end{equation}

\begin{equation} \bar{G}_2 = \left( \begin{array}{cc} \alpha & \beta \\ \gamma & \delta 
\end{array} \right) \hspace{5mm} \Rightarrow \hspace{5mm} K_1\bar{G}_2K_1^{-1} = 
\left( \begin{array}{cc} \alpha & \beta x_1^{-1} \\ \gamma x_1 & \delta 
\end{array} \right) \ . \end{equation}

\noindent On the other hand 

\begin{equation} 
U_{G_1}U_{G_2} =  U_{\bar{G}_1} g_1 (U_{\bar{G}_2} g_2) =  U_{\bar{G}_1}g_1(U_{\bar{G}_2})g_1g_2 \ . 
\end{equation}

\noindent Now it follows from eq. (\ref{rotp}) that $g_1(e^{i\theta}\sqrt{p}) 
= e^{i\theta}\sqrt{p}$ if $x_1$ is a quadratic residue, and $g_1(e^{i\theta}\sqrt{p}) 
= -e^{i\theta}\sqrt{p}$ if it is not. In the latter case $(-\beta x_1^{-1}|p) = 
- (-\beta |p)$. It then follows by inspection of eqs. (\ref{1}) that 

\begin{equation} g_1(U_{\bar{G}_2}) = U_{K_1\bar{G}_2K_1^{-1}} \ . \end{equation}

\noindent Given that the representation of $\SL$ is faithful \cite{Neuhauser, Appleby}, 
the representation of $\GL$ is faithful too.

The sense in which the operators we are dealing with are g-unitary, as opposed 
to merely g-linear, was spelt out in ref. \cite{AYAZ}. Suppose $G = \bar{G}K_{\Delta}$, where 
det$\bar{G} = 1$. The adjoint of $U_G$ is then defined by 

\be
\la U_G \psi, \phi\ra = g_{\Delta}^{-1} \bigl( \la \psi, U^{\dagger}_G\phi \ra\bigr)
\ . \ee

\noindent An explicit expression for the adjoint is 

\be
U^{\dagger}_G = g_{\Delta}^{-1} (U_{\bar{G}}^{\dagger}) g_{\Delta}^{-1} \ . 
\ee

\noindent This expression is readily seen to imply that 
$U^{\dagger}_G U_G=U_G U^{\dagger}_G = {\bf 1}$. 

The action of the group is restricted to those vectors in Hilbert space 
whose components belong to the cyclotomic field $\fd{Q}(\omega)$. While this is a severe 
restriction, it does include the vectors in the MUB, since they form an 
orbit under the Clifford group. It also includes every vector that can be 
reached from the computational basis by using a larger set of transformations 
known as the Clifford hierarchy. This set is large enough for the purposes 
of universal quantum computation \cite{Gottesman}. 

The g-unitaries will preserve the Hilbert space norm only if this norm is 
rational, which it may well not be. This means that their action is wildly 
discontinuous in general. 
Thus, consider the transformation $\omega \rightarrow \omega^2$ for $d = 5$, 
and its action on the vector 

\begin{equation} \left( \begin{array}{c} \omega^2 + \omega^3 \\ 1 \\ 1 \\ 1 \\ 1 
\end{array} \right) \hspace{5mm} \rightarrow \hspace{5mm} 
\left( \begin{array}{c} \omega^4 + \omega \\ 1 \\ 1 \\ 1 \\ 1 
\end{array} \right) \ . \end{equation}

\noindent These vectors are real, and can be approximated by rational vectors 
which are left invariant by the g-unitary, while the vector which is approximated 
is moved a long distance in Hilbert space. This behaviour should be kept in mind. 

A g-unitary operator does preserve a norm which is obtained by multiplying the 
scalar product of two vectors with all its $d-2$ Galois conjugates \cite{Janusz}. 
However, if 
this norm has a physical meaning it is hidden from us.

\section{The g-extended Clifford group in the prime power case}

We hope that the idea of g-unitaries is by now clear, and turn to the complications 
that occur in prime power dimensions. We will assume the material in section 3. 

When $d = p^n$ the elements of the Heisenberg group can be labelled by elements 
of the finite field $\fd{F}_d$ of order $p^n$. Thus we write \cite{Appleby}  

\begin{equation} \omega = e^{\frac{2\pi i}{d}} \ , \hspace{5mm} 
Z_u|x\rangle = \omega^{\tr(xu)}|x\rangle \ , \hspace{5mm} 
X_u|x\rangle = |x+u\rangle \ , 
\end{equation} 

\noindent with $x, u \in \fd{F}_d$ and the field theoretic trace $\tr u$ of an element 
of the finite field lies in the ground field 
$\fd{F}_p$. Although this is not immediately obvious the resulting group is isomorphic 
to the direct product of $n$ copies of the Weyl-Heisenberg group in dimension $p$ 
\cite{Vourdas, Appleby}. The displacement operators are 

\be
D_{\mbf{u}} |x\ra = \omega^{\frac{1}{2} \tr(u_1u_2) +\tr(u_2x)}|x+u_1\ra ,
\ee

\noindent where the field theoretic trace appears again. Complete sets of MUB 
again exist \cite{Fields}. They arise as eigenbases of maximally abelian subgroups 
with only the unit element in common just as they do in prime dimensions \cite{Vatan}, 
but now there are many options for how to do this. 

Two different Clifford groups can be defined \cite{Gross}. The one 
we are interested in here is a subgroup of the other, and it has been 
called the restricted Clifford group \cite{Appleby}. 
It leaves a given complete set of MUB invariant, and includes symplectic unitaries 
representing the $\SL(2,\mathbb{F}_d)$ matrices

\be
G = \bmt \alpha & \beta \\ \gamma & \delta \emt \ . 
\ee

\noindent A faithful representation is given by \cite{Appleby}

\be
U_G 
=
\begin{cases}
l(\alpha) \sum_{x\in\mathbb{F}_d} \omega^{\tr(\frac{\alpha \gamma}{2} x^2)} 
|\alpha x\rangle \langle x | \qquad & \text{if $\beta = 0$}
\\ \vphantom{\Biggl(}
\frac{\tilde{l}(-\beta)}{\sqrt{d}} \sum_{x,y\in \mathbb{F}_d} 
\omega^{\tr\left( \frac{1}{2\beta} (\alpha y^2 - 2 x y + \delta x^2)\right)} 
|x\rangle \langle y| \qquad & \text{if $\beta \neq 0$} \end{cases}
\label{eq:symUDef}
\ee

\noindent where 

\be
\tilde{l}(x) = \begin{cases} -i^{-\frac{n(p+3)}{2}} l(x) \qquad & x\neq 0 \\ 1 \qquad & x=0 \end{cases}  
\ee

\noindent and $l(x)$ is the quadratic character of $\fd{F}_d$, equal to $+1$ if 
$x$ can be written as a square, and equal to $-1$ otherwise. 

By adjoining the matrix $K_x$, given in eq. (\ref{Kx}), we can extend the representation 
to include the group $\GL_p(2, \fd{F}_d)$ consisting of two-by-two matrices whose 
determinants are non-zero and are in the ground field $\fd{F}_p$. For any $x \in \fd{F}_d$ the 
matrix $K_x$ is represented in the same way as in section 5. Hence the case of prime 
power dimension differs from the prime dimensional case in that the g-extended Clifford 
group includes only a proper subgroup of $\GL$.

We know that the representation of $\SL$ is faithful \cite{Neuhauser, Appleby}, and 
given that we can prove that the representation of $\GL$ is faithful too:

\begin{theorem}
The group $\GL_p(2,\fd{F}_d)$ is faithfully represented by eqs. 
(\ref{eq:symUDef}) provided that $d = p^n$ and $n$ is odd.
\end{theorem}
\begin{proof} 
First recall the basic fact that 
$\fd{F}_p = \{x\in \fd{F}_d \colon x^p = x\}$. Let $\theta$ be a primitive element 
of $\fd{F}_d$. Then it is not difficult to show that $\theta_p = 
\theta^{1+p+\dots + p^{n-1}}$ belongs to $\fd{F}_p$, and is in fact a primitive element 
of $\fd{F}_p$ (since this is true for all non-zero elements of $\fd{F}_p$). 
%

Let $G_1$, $G_2$ be arbitrary elements of $\GL_p(2,\fd{F}_d)$.  We write $G_{1,2} = 
\bar{G}_{1,2}K_{\Delta_{1,2}}$, where $\bar{G}_{1,2} \in \SL(2,\fd{F}_d)$. Then 

\be
U_{G_1} U_{G_2} = U_{\bar{G}_1} g_{\Delta_1}(U_{\bar{G}_2}) g_{\Delta_1} g_{\Delta_2} \ . 
\ee

\noindent We know that $\Delta_1 = \theta_p^{u}$ for some integer $u$.  It is easily seen that $g_{\Delta_1} = g_{\theta_p}^{u}$, and that $\Delta_1$ is a quadratic residue in $\fd{F}_p$  if and only if $u$ is even.  Applying eq. (\ref{eq:symUDef}) to

\be
\bar{G}_2 = \bmt \alpha & \beta \\ \gamma & \delta\emt
\ee
we have
\be
U_{\bar{G}_2} =
\begin{cases}
l(\alpha) \sum_{x\in \fd{F}_d} \omega^{\frac{1}{2} \tr(\alpha \gamma x^2)}|\alpha x\ra\la x| \qquad &\beta = 0
\\
\frac{\tilde{l}(-\beta)}{\sqrt{d}} \sum_{x,y\in \fd{F}_d} \omega^{\frac{1}{2} \tr(\beta^{-1}(\alpha y^2-2xy+\delta x^2)}|x\ra\la y |
\qquad & \beta \neq 0 \ . 
\end{cases}
\ee
Since
\ea{
\frac{\tilde{l}(-\beta)}{\sqrt{d}}& \notin \fd{Q} & 
\left(\frac{\tilde{l}(-\beta)}{\sqrt{d}}\right)^2 & \in \fd{Q}
}
and since $g_{\theta_p}$ generates the Galois group, we must have
\be
g_{\theta_p} \left(\frac{\tilde{l}(-\beta}{\sqrt{d}}\right) = - \frac{\tilde{l}(-\beta)}{\sqrt{d}} \ ,
\ee
implying
\be
g_{\Delta_1} \left(\frac{\tilde{l}(-\beta)}{\sqrt{d}}\right) =(-1)^{u} \frac{\tilde{l}(- \beta)}{\sqrt{d}} \ .
\ee
In view of  Lemma 1 in ref. \cite{Appleby}, and our assumption that $n$ is odd, we may write this in the form
\be
g_{\Delta_1} \left(\frac{\tilde{l}(-\beta)}{\sqrt{d}}\right) = \frac{\tilde{l}(- \Delta_1 \beta)}{\sqrt{d}} \ .
\ee
Therefore

\begin{equation}
g_{\Delta_1}(U_{\bar{G}_2}) =
\begin{cases}
l(\alpha) \sum_{x\in \fd{F}_d} \omega^{\frac{1}{2} \tr(\Delta_1 \alpha \gamma x^2)}|\alpha x\ra\la x| \qquad &\beta = 0
\\
\frac{\tilde{l}(-\Delta_1 \beta)}{\sqrt{d}} \sum_{x,y\in \fd{F}_d} \omega^{\frac{1}{2} \tr(\Delta_1\beta^{-1}(\alpha y^2-2xy+\delta x^2)}|x\ra\la y |
\qquad & \beta \neq 0 \ . 
\end{cases} \end{equation}

\noindent Thus 

\begin{equation} g_{\Delta_1}(U_{\bar{G}_2}) = U_{K_{\Delta_1} \bar{G}_2 K_{\Delta_1}^{-1}} 
\ . \end{equation}

\noindent In view of the fact that the representation of $\SL(2,\fd{F}_d)$ is faithful 
we can now deduce

\be
U_{G_1} U_{G_2} = U_{G_1G_2}.
\ee
implying that  the representation of $\GL_p(2,\fd{F}_d)$  is also faithful, as claimed.
\end{proof}

As an immediate consequence of this one has

\be
U^{\dagger}_G = U_{G^{-1}}
\ee

\noindent for all $G \in \GL_p (2, \fd{F}_d)$. We also remark that if $n$ is even 
then it follows from the above that 

\be
U_{G_1}U_{G_2} = \pm U_{G_1G_2}
\ee

\noindent so that the representation is, in a sense, ``close to faithful''.

It follows from Lemma 1 of ref. \cite{Appleby} that $\fd{F}_p$ contains numbers which are quadratic non-residues with respect to the embedding field $\fd{F}_d$ if and only if $n$ is odd.  Consequently, extending the Clifford group to include the full set of $g$-unitaries will give us the ``missing'' M\"{o}bius transformations discussed in Section $4$ if and only if $n$ is odd.

\section{Complementarity polytopes}
To give a geometrical interpretation to the g-unitaries we place ourselves 
in the set of Hermitean matrices of unit trace, regarded as an Euclidean space 
equipped with the standard Hilbert-Schmidt metric. The set of density matrices 
forms a convex body within this space, which we call Bloch space. Usually one 
thinks of Bloch space as a $(d^2-1)$-dimensional vector space, with its origin 
at the maximally mixed state. Anyway the distance between two 
Hermitean matrices $M_1$ and $M_2$ (be they density matrices or not) is 

\begin{equation} D(M_1,M_2) = \sqrt{\frac{1}{2} \mbox{Tr}(M_1-M_2)^2} \ . \end{equation} 

\noindent In this Bloch space we introduce a regular 
simplex with $d^2$ vertices represented by matrices $A_{ij}$, with each index ranging 
over $d$ possible values. They are chosen to obey  

\begin{equation} \mbox{Tr}A_{ij}A_{kl} = \left\{ \begin{array}{l} d \ \ 
\mbox{if} \ (i,j) = (k,l) \\ 0 \ \ \mbox{otherwise .} \end{array} \right. \end{equation}

\noindent At the end things will be arranged so that these matrices can be identified 
with Wootters' phase point operators \cite{Wootters}, but at the outset they just 
define the vertices of a regular simplex in Bloch space. The simplex is centred at 
the maximally mixed state by insisting that 

\begin{equation} \sum_{i=0}^{d-1}\sum_{j=0}^{d-1}A_{ij} = d{\bf 1} \ . 
\label{summa} \end{equation}

\noindent A facet of the simplex consists of all matrices $M$ such that Tr$A_{ij}M = 0$ 
for some fixed choice of $i,j$. 

Now suppose that $d$ is a prime or a prime power, in which case there exists a combinatorial 
structure known as a finite affine plane with $d^2$ points. One 
defines subsets of points, known as lines, such that any pair of points belong 
to a unique line and such that for every point not belonging to a line there 
is a unique line disjoint from the given line and containing the given point. 
(This is the parallel postulate in affine geometry.) The most important theorems 
concerning finite affine planes state that \cite{Bennett}

\begin{enumerate}
\item the number of lines in the finite affine plane equals $d(d+1)$,
\item each line contains $d$ points,
\item each point is contained in $d+1$ lines,
\item and there are altogether $d+1$ sets of $d$ disjoint lines.
\end{enumerate}

\noindent Therefore we can associate a unit trace operator (and hence a point in 
Euclidean space) to each given line by summing the $d$ phase point operators 
contained in the line, 

\begin{equation} P^{(z)}_r = \frac{1}{d}\sum_{{\bf p} \in l^z_r} A_{\bf p} \ . \label{sumrest} \end{equation}

\noindent The index $z \in \{ \fd{F}_d \cup \infty \} $ labels the $d+1$ pencils 
of parallel lines with the symbol $\infty$ used 
to label the ``vertical'' pencil, and the index $r \in \fd{F}_d$  
labels the individual lines within such a pencil. There are $d(d+1)$ such line operators 
altogether.

Here we have tacitly identified 
the affine plane with the vector space $(\fd{F}_d)^2$. Actually affine planes not 
coordinatized by finite fields do exist and the use of a finite field for 
labelling purposes is not mandatory, but we will stick to it. 

Due to eq. (\ref{summa}) and the fact 
that $d+1$ lines intersect in each point eq. (\ref{sumrest}) can be inverted to give 

\begin{equation} A_{\bf p} = \sum_{l^z_r \ni {\bf p}} P^{(z)}_r - {\bf 1} \ . \label{restsum} \end{equation}

\noindent The sum is over all the $d+1$ operators representing lines going 
through the point represented by $A$.  

\begin{figure}
\begin{center}
\includegraphics[width=.8\textwidth]{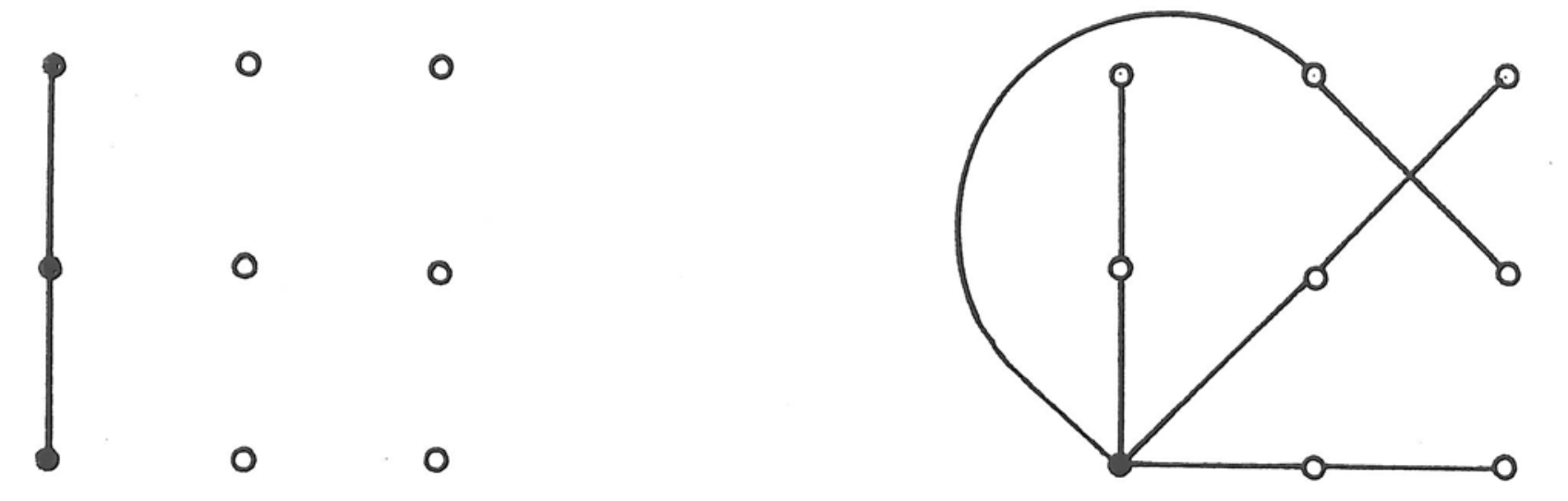}
\end{center}
\caption{\small A finite plane is formed by the operators $A_{ij}$. To the left 
we illustrate the equation $P_0^{(\infty )} = \frac{1}{3}(A_{00} 
+ A_{01} + A_{02})$. To the right we illustrate the equation $A_{00} = P^{(0)}_0 + 
P^{(1)}_0 + P^{(2)}_0 + P^{(\infty )}_0 - {\bf 1}$. Note that both the points and 
the lines in the affine plane correspond to points in Bloch space.}
\label{fig:2}
\end{figure}

Using the combinatorics of the affine plane one easily checks that 

\begin{equation} \mbox{Tr}P^{(z)}_rP^{(z')}_{r'} = \left\{ \begin{array}{lll} 1 & 
\mbox{if} \ z=z' \ \mbox{and} \ r = r' & \mbox{(identical)} \\ 
0 & \mbox{if $z = z'$ and $r \neq r'$} & \mbox{(no common point)} \\ \frac{1}{d} & \mbox{if} 
\ z \neq z' & \mbox{(one common point)} \ . \end{array} \right. \end{equation}  

\noindent The geometrical meaning of this, in Euclidean space, is that any collection 
of $d$ parallel 
(non-intersecting) lines forms a regular simplex spanning a $(d-1)$-dimensional 
plane, and that the collection of $d+1$ planes defined in this way are 
totally orthogonal. The convex hull of the $d(d+1)$ vertices defines the 
complementarity polytope \cite{Asa}. 

A complementarity polytope exists in all Euclidean spaces of dimension 
$d^2-1$. What is special about $d$ being a prime power is only that we were able 
to use the combinatorics of a finite affine plane to inscribe it in a 
regular simplex. To understand its full face structure we define generalized 
phase point operators by picking one 
matrix $P^{(z)}_r$ for each value of $z$, and summing them to obtain  

\begin{equation} A_{\vec{r}} = \sum_z
P^{(z)}_r - \frac{1}{d}{\bf 1} \ . 
\label{ppo} \end{equation}

\noindent Unlike the summation in eq. (\ref{restsum}), we are now allowing 
arbitrary choices of lines (label $r$) from the pencils (label $z$). 
The operators $A_{\vec{r}}$ have unit trace. There are $d^{d+1}$ such operators altogether labelled by a vector $\vec{r}$ with entries in $\fd{F}_d$. 
It is easy to see that 

\begin{equation} \mbox{Tr}A_{\vec{r}}P^{(z)}_r = 
\left\{ \begin{array}{ll} 1 & \mbox{if the $z$th component of} \ \vec{r} \ 
\mbox{equals} \ r \\ 
0 & \mbox{otherwise} \ . \end{array} \right. \end{equation}

\noindent Therefore the entire complementarity polytope is confined between 
pairs of parallel hyperplanes containing pairs of orthocomplemented faces. 
This includes faces that contain $d^2-1$ 
vertices spanning a facet associated to the phase point operator $A_{\vec{r}}$, 
and all facets arise in this way. The particular phase point operator simplex 
we started out with is just one out of $d^{d-1}$ such simplices in which the 
complementarity polytope is inscribed, but we will soon see why it is useful 
to single out one of them for attention.

If $d = 2$ Bloch space has $2^2-1 = 3$ dimensions and the complementarity polytope 
is a regular octahedron with $2\cdot 3$ vertices and $2^3$ facets, inscribed in 
a regular tetrahedron. However, we will see that the even and odd dimensional cases 
differ a bit from the point of view of their symmetry groups. 

\begin{figure}
\begin{center}
\includegraphics[width=.5\textwidth]{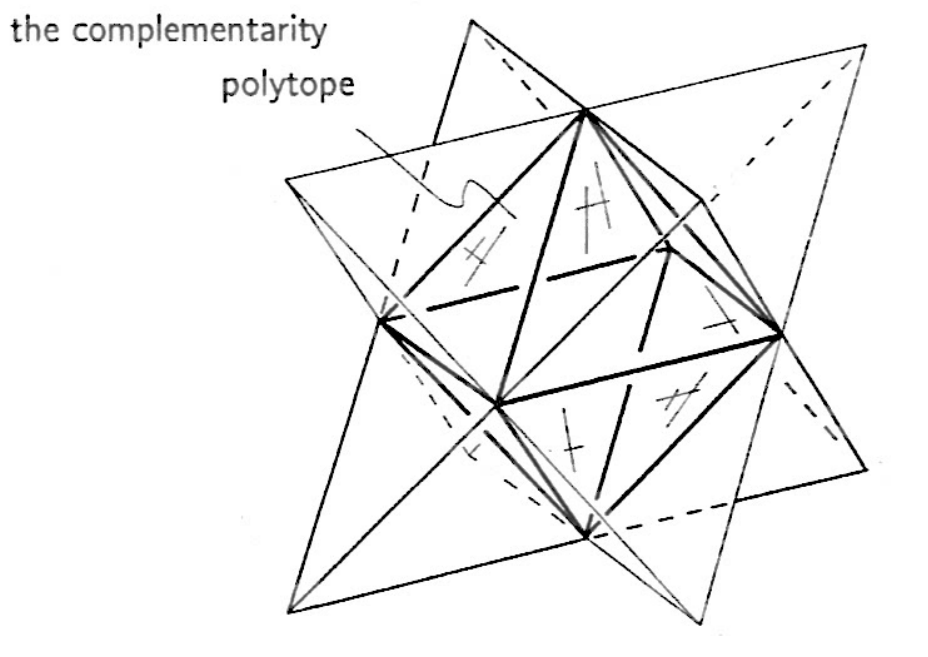}
\end{center}
\caption{\small When $d = 2$ the complementarity polytope is inscribed 
in $2^{2-1}= 2$ simplices. It is a special property of $d = 2$ 
that these are on the same footing relative to the set of quantum states 
(the round Bloch sphere). Reprinted with permission \cite{ericsson}.}
\label{fig:1}
\end{figure}

\section{The symmetry group of the complementarity polytope}
The symmetry group of the complementarity polytope is the huge group 

\begin{equation} S_{d+1}\times S_d\times S_d \times \dots \times S_d \in 
O(d^2-1) \ , \label{grupp3} \end{equation}

\noindent where $S_d$ is the group of all permutations of the vertices of a 
$(d^2-1)$-dimensional simplex, and $S_{d+1}$ is the group of permutations of the 
$(d-1)$-planes. However, we will naturally insist that its vertices correspond to 
pure quantum states, which means that the polytope must be inscribed into the convex 
body of quantum states. Given that we have a complete set of MUB available this is 
achieved by 

\begin{equation} P_r^{(z)} = \Pi^{(z)}_r = |e^{(z)}_r\rangle \langle e^{(z)}_r| \ 
. \label{anchor} \end{equation}

\noindent This is the tricky step, but we know it can be done if $d$ is a 
power of a prime number. 
In any case the symmetry group of the body of 
quantum states is, ignoring reflections,

\begin{equation} U(d)/U(1) \in SO(d^2-1) \ . \end{equation}

\noindent The intersection of the two groups is precisely the Clifford group modulo 
phase factors. Provided that $d$ is an odd prime power this group is in fact 

\begin{equation} \SL(2,\fd{F}_d) \ltimes (\fd{F}_d)^2 \ , \label{grupp2} \end{equation}

\noindent where $\fd{F}_d$ is considered as a group under addition. 
Including anti-unitary symmetries leads to the extended Clifford group, 
but if $p > 3$ we have to include g-unitaries as well in order to have the 
full projective group of M\"obius transformations acting on the label $z$. 
We would like to understand in geometrical terms what property of the complementarity 
polytope singles out this group for attention.

In the preceding section we showed that a complementarity polytope can be 
inscribed into a regular simplex (whose facets are suitably selected facets of 
the complementarity polytope) provided an affine plane of order 
$d$ exists. The full symmetry group of the simplex including reflections 
is the symmetric group $S_{d^2}$. 
We are interested in the largest common subgroup of the symmetry groups of 
the two polytopes. Provided $d$ is a prime power, and provided the finite 
affine plane is coordinatized by a finite field $\fd{F}_d$ of order $d$, the 
common subgroup is easily recognized. In the affine plane it must take points 
to points and lines to lines. This means that it is isomorphic to the 
affine group over the field in question, namely  

\begin{equation} \GL(2,\fd{F}_d)\ltimes (\fd{F}_d)^2 \ , \label{grupp1} \end{equation}

\noindent This is almost, but not quite, the same answer as that obtained 
when one restricts oneself to those symmetries of the complementarity 
polytope that preserve also the inscribed body of density matrices, but if 
$d> 3$ it is not quite the same. When $d$ is prime the g-unitaries provide 
all of the extra transformations. When $d$ is a prime power they provide 
some of the extra transformations, namely those coming from the subgroup 
$\GL_p$ of $\GL$. 

Thus we have arrived at our geometric interpretation of the g-unitaries: 
When their action is restricted to the MUB vectors they are in fact rigid 
rotations in Bloch space. Harking back to the end of section 5 we note that 
this interpretation hinges on the fact that all the trace inner 
products between these vectors are rational numbers. Hence the interpretation 
of g-unitaries in terms of rotations in Bloch space has a very limited scope, 
and does not apply to their action on arbitrary vectors in their domain.

We note that if $d$ is even there is a subtle difference. Consider $d = 2$. 
Then the two groups $\GL$ and $\SL$ coincide, and the 
group (\ref{grupp1}) is isomorphic to $S_4$, the full symmetry group of the 
tetrahedron including reflections. Quantum mechanically this is realized by 
a combination of unitary and anti-unitary operations. The Clifford group 
modulo phases equals the symmetry group of the octahedron, which is again 
isomorphic to $S_4$ but subtly different from the group (\ref{grupp2}) \cite{DMA}. 
Quantum mechanically it is realized by unitary operations, not all of which 
preserve the phase point operator simplex. 

Concerning odd prime $d$ we observe that the projective group $\PGL$ is 
a subgroup of the $S_{d+1}$ factor of the group (\ref{grupp3}), which permutes  
the bases in the set of MUB. Moreover $\PGL$ has two components depending on 
whether the determinant $\Delta$ of the $\GL(2)$ matrix is a quadratic 
residue or not. If $p = 3$ modulo 4 it happens that $-1$ is a quadratic 
non-residue, which means that the full set of projective transformations 
is recovered from the extended symplectic group (which is realized by 
unitary and anti-unitary transformations). If on the other hand $p = 1$ 
modulo 4 then -1 is a quadratic residue, and one needs to consider 
general g-unitaries to recover the full set of projective transformations. 
For prime power dimensions $d = p^n$ the situation is more complicated, and 
we do not obtain all of the projective transformations from the g-unitaries 
unless $n$ is odd.

We also note that, given the identification in eq. (\ref{anchor}), there 
exists a special set of phase point operators taking a simple form. When $d$ 
is odd this includes the parity operator 

\begin{equation} A = \frac{1}{d}\sum_{\bf u}D_{\bf u} \ , \label{ppA} \end{equation}

\noindent where the sum extends over the finite phase space. The full set of special 
phase point operators consists of 

\begin{equation} A_{\bf p} = D_{\bf p}AD^{-1}_{\bf p} = \sum_{\bf u} 
\omega^{2\Omega({\bf p},{\bf u})}D_{\bf u} \ . \end{equation}

\noindent Thus the special phase point operator simplex is a unitary operator 
basis obtained in 
a simple way from the unitary operator basis provided by the Weyl-Heisenberg 
group. The generalized phase point operators defined by eq. (\ref{ppo}) can 
also be collected into simplices, transforming into each other under the 
symplectic group \cite{Gibbons}, but their spectral properties are not as 
attractive \cite{Subhash, Galvao}.    

\section{MUB-cyclers}
\label{sec:MUBcyclers}
We now come to the question of whether the projective permutations of the bases include MUB-cyclers, that is transformations that cycle through all bases in succession. Thus we need an element of $\PGL$ of order $d+1$, or equivalently an element of $\GL$ of suborder $d+1$ (the definition of suborder will be elaborated in the next paragraph). In odd prime dimensions the existence of such an element $G$ is sufficient, but in the case of odd prime power dimensions $d=p^n$, where entries of $G$ belong to the field $\fd{F}_{p^n}$ we need an extra condition, that the determinant of $G$ belongs to the prime field $\F_p$ (i.e. it is an integer modulo $p$) so that $G$ admits a g-unitary representation $U_G$ as in section 5. In the following theorems, we establish that MUB-cycling g-unitaries exist if and only if the exponent $n$ is odd, and provide a characterization of every possible MUB-cycler for when $n$ is odd.

It follows from the expression for the M\"{o}bius transformation (\ref{Mobius}) that if
\begin{equation}
G = \begin{pmatrix}\alpha & \beta \\ \gamma &\delta \end{pmatrix}, \text{ }
G^m = \begin{pmatrix}\alpha_m & \beta_m \\ \gamma_m &\delta_m \end{pmatrix}
\end{equation}
then $U_{G^m}$ takes $z = 0$ to $z = \beta_m/\delta_m$ (respectively $\infty$) if $\delta_m \ne 0$ (respectively $\delta_m = 0$). So if we define the suborder of $G$ to be the smallest positive integer $m$ for which $\beta_m= 0$, then $U_{G^m}$ will take a MUB basis back to itself. However, as $U_{G^m}$ can permute the vectors in the basis, $m$ need not be the order of $G$. Generally, the suborder is a factor of the order of $G$. Following Lemma \ref{lem1}, we will see that the smallest positive integer $m$ such that $G^m$ is proportional to the $2 \times 2$ identity matrix is an equivalent definition of the suborder of $G$.

Let $G$ be an element of $\GL(2,\fd{F}_d)$ with trace $t$ and determinant $\Delta \ne 0$. The eigenvalues of $G$ are roots of the characteristic polynomial $x^2-tx+\Delta = 0$ and are  given by
\begin{equation}
\lambda_{\pm} = (t \pm \sqrt{t^2-4\Delta})/2.
\end{equation}
If $t^2-4\Delta$ is zero or a quadratic residue, i.e. it has a non-zero square root in $\F_d$, then $\lambda_{\pm}$ belong to the field $\F_d$. Otherwise, the eigenvalues do not belong to $\F_d$, but they are still well-defined, and we can extend the field $\F_d$ to $\F_{d^2}$ to include them. To deal with these cases, it is convenient for us to classify $\GL$ elements into three types, as summarized in the table below.
\begin{table}[h]
\centering
\begin{tabular}{|c|l|c|}
\hline
 & Definition in terms of $t$ and $\Delta$ & Equivalent definition\\ \hline
Type 1 & $t^2-4\Delta$ is a quadratic residue  & $\lambda_{\pm} \in \F_d, \lambda_{+} \ne \lambda_{-}$\\ \hline
Type 2 & $t^2-4\Delta$ is a quadratic non-residue & $\lambda_{\pm} \notin \F_d, \lambda_{+} \ne \lambda_{-}$ \\ \hline
Type 3 & $t^2-4\Delta = 0$ & $\lambda_{\pm} =t/2 \in \F_d$ \\ \hline
\end{tabular}
\caption{A classification of $\GL$ elements into three types, among which only type 2 can give rise to MUB-cyclers, as will be seen in the following theorems.}
\label{tbtype}
\end{table}

Throughout this section, we will assume that the dimension $d$ is a prime power of the form $d=p^n$, where $p$ is an odd prime number.

\begin{theorem}\label{thm1}
Let $G$ be an element of $\GL(2,\F_d)$ with determinant $\Delta$. 
\begin{enumerate}[\hskip 4mm 1.]
\item If $G$ is of type 1, then $G$ has suborder of at most $d-1$.
\item If $G$ is of type 2, then $G$ has suborder of at most $d+1$ and satisfies
\begin{equation}\label{eqthm1}
G^{d+1} = \Delta I \text{\hskip 6mm ($I$ is the $2 \times 2$ identity matrix).}
\end{equation}
\item If $G$ is of type 3, then $G$ has suborder of at most $d$.\end{enumerate}
\end{theorem}

We will start with the following lemma, whose proof by induction can be carried out straightforwardly. We will leave the proof for our readers.

\begin{lemma}\label{lem1}
Let $A$ be any $2\times2$ matrix, with trace $t$ and determinant $\Delta$. For any integer $m \ge 1$, it holds that
\begin{equation}\label{eqAm}
A^m = s_m A - s_{m-1} \Delta I,
\end{equation}
where the sequence $\{s_m\}$ is defined by the recurrence relation
\begin{equation}\label{eqsmr}
s_{m+1} = t s_m - \Delta s_{m-1},
\end{equation}
with $s_0=0$ and $s_1=1$. Equivalently, $s_m$ can be calculated by
\begin{equation}\label{eqsm}
s_m = \begin{cases} (\lambda_+^m - \lambda_-^m)/(\lambda_+-\lambda_-) \qquad & \text{if $\lambda_+ \ne \lambda_-$} \\ 
m\lambda_+^{m-1} \qquad & \text{if $\lambda_+ = \lambda_-$} \end{cases}
\end{equation}
where $\lambda_{\pm}$ are roots of the characteristic polynomial $x^2 - t x + \Delta$.
\end{lemma}

\begin{remark}
If $A = \begin{pmatrix}\alpha & \beta \\ \gamma &\delta \end{pmatrix}$, we can explicitly rewrite Eq. (\ref{eqAm}) as 
\begin{equation}
A^m = \begin{pmatrix}
s_m\alpha-s_{m-1}\Delta & s_m\beta \\
s_m\gamma & s_m\delta - s_{m-1}\Delta
\end{pmatrix}	
\end{equation}
and see that if $\beta \ne 0$ (for a non-zero determinant, we can always make $\beta \ne 0$ using the form in Eq. (\ref{defcanonical}) at the end of this section) the suborder of $A$ is the smallest positive integer $m$ for which $s_m=0$. $A^m$ is then proportional to the identity matrix.\\
\end{remark}

\begin{myproof}[Proof of Theorem \ref{thm1}] 

Let $\lambda_{\pm}$ be the eigenvalues of $G$, based upon which we define a sequence $\{s_m\}$ just as in (\ref{eqsm}). Note that although $\lambda_{\pm}$ might not be in the field $\F_d$, the sequence $\{s_m\}$ always is, as can be seen from the recursive definition in (\ref{eqsmr}). Lemma~\ref{lem1} implies that if $s_m = 0$ for some $m$, then $G^m = -s_{m-1}\Delta I$, and therefore the suborder of $G$ is at most $m$. Let us now consider specific cases. Facts about finite fields in Section \ref{secfield} will be used implicitly here.
\begin{enumerate}[\hskip 4mm 1.]
\item If $G$ is of type 1, then the eigenvalues $\lambda_{\pm}$ are in $\F_d$, and therefore
\begin{equation}
\lambda_{+}^{d-1} = \lambda_{-}^{d-1} = 1,
\end{equation}
which implies
\begin{equation}
s_{d-1} =(\lambda_{+}^{d-1}-\lambda_{-}^{d-1})/(\lambda_{+}-\lambda_{-}) = 0.
\end{equation}
Therefore $G$ has suborder of at most $d-1$.
\item If $G$ is of type 2, we create an extension field $\F_{d^2}$ from the base field $\F_d$ and the generator $j \equiv \sqrt{t^2 - 4 \Delta}$. Since $(j^d)^2=(t^2-4\Delta)^d = t^2-4\Delta=j^2 $, we have $j^d = \pm j$. Because $j$ is not in the field $\F_d$ we cannot have $j^d = j$, and it therefore must be the case that $j^d = -j$. As $d$ is odd, we have
\begin{equation}\label{eqlambdad}
\lambda_{\pm}^d 
= \frac{\left(t \pm j \right)^d}{2^d} 
= \frac{t \pm j^d}{2}
=\frac{t \mp j}{2} = \lambda_{\mp}.
\end{equation}
We then use (\ref{eqsm}) to derive
\begin{equation}
s_{d}= (\lambda_{+}^{d}-\lambda_{-}^{d})/(\lambda_{+}-\lambda_{-}) = -1,
\end{equation}
\begin{equation}
s_{d+1}= \frac{\lambda_{+}^{d+1}-\lambda_{-}^{d+1}}{\lambda_{+}-\lambda_{-}} 
= \frac{\lambda_{+}\lambda_{-} -\lambda_{-}\lambda_{+}}{\lambda_{+}-\lambda_{-}} 
= 0,
\end{equation}
and therefore
\begin{equation}
G^{d+1} = s_{d+1}G - s_{d}\Delta I = \Delta I.
\end{equation}
It follows that $G$ has suborder of at most $d+1$.
\item
If $G$ is of type 3, then $\lambda_{\pm} = t/2$. It follows from Eq. (\ref{eqsm}) that  $s_d = d\lambda_{+}^{d-1} = 0$, so $G$ has suborder of at most $d$.
\end{enumerate}
\end{myproof}

\begin{lemma}\label{lem2}
Let $G \in \GL_p(2,\F_d)$, i.e. an element $\GL(2,\F_d)$ whose determinant $\Delta$ is in $\F_p$. Let $\bar{\theta}$ be a primitive element of $\F_{d^2}$. Note that $(d-1)/(p-1)$ is an integer, so we can define an element $\eta \in \F_{d^2}$ as
\begin{equation}\label{defeta}
\eta \equiv \bar{\theta}^{(d-1)/(p-1)}.
\end{equation}
Then $G$ is of type 2 if and only if it has eigenvalues $\eta^r$ and $\eta^{dr}$, for some integer $r$ not a multiple of $(d+1)/2$ in the range $0 < r < (p-1)(d+1)$.

\end{lemma}

\begin{proof}
Assume that $G$ is of type 2, and let $\lambda_{\pm}$ be its eigenvalues. Following Eq. (\ref{eqlambdad}) in the proof of Theorem \ref{thm1} we have $\lambda_{\pm}^d = \lambda_{\mp}$, so we may write
\begin{equation}
\lambda_{+} = \bar{\theta}^k \hskip 15mm \lambda_{-} = \bar{\theta}^{dk}
\end{equation}
for some integer $k$ in the range $1 \le k \le d^2-1$. The fact that the eigenvalues are not in $\F_d$ means that $k$ is not a multiple of $d+1$. The fact that $\Delta \in \F_p$ means
\begin{equation}
\bar{\theta}^{pk(d+1)} = \Delta^p = \Delta = \bar{\theta}^{k(d+1)}
\end{equation}
So
\begin{equation}\label{eqthetak}
\bar{\theta}^{k(p-1)(d+1)} = 1
\end{equation}
implying that $(d-1) \mid k(p-1)$. Let $r = k(p-1)/(d-1)$. Then the eigenvalues can be written as
\begin{equation}
\lambda_{+} = \eta^r \hskip 15mm \lambda_{-} = \eta^{dr}.
\end{equation}
The requirement that $\lambda_{\pm} \notin \F_d$ means 
\begin{equation}
\eta^{dr} = \bar{\theta}^{rd(d-1)/(p-1)} \notin \F_d
\end{equation} 
which is true if and only if $r$ is not a multiple of $(d+1)/2$.\\

Conversely, if $G$ has eigenvalues of the form $\lambda_{+} = \eta^r$ and $\lambda_{-}=\eta^{dr}$, where $r$ is not a multiple of $(d+1)/2$, then $\lambda_{\pm}$ are not in the field $\F_d$, and $G$ is therefore of type 2. One can further verify that its trace is in $\F_d$ and its determinant is in $\F_p$ by defining 
\begin{equation}
t \equiv \eta^r + \eta^{dr} \hskip 15mm \Delta \equiv \eta^{(d+1)r}
\end{equation}
and using the facts $\eta^{d^2}=\eta$ and $\eta^{(d+1)p}=\eta^{d+1}$ to check that
\begin{equation}
t^d = t \hskip 15mm \Delta^p = \Delta.
\end{equation}

\end{proof}
With Lemma \ref{lem2}, all type-2 elements of $\GL(2,\F_d)$ whose determinant $\Delta$ is in $\F_p$ (this extra condition is to   guarantee the feasibility of their g-unitary representation) can now be characterized by an integer $r$, via their eigenvalues. In the next theorem, we will pin down which exact values of $r$ correspond to MUB-cyclers.\\
\begin{theorem}\label{thm2}
Let $G \in \GL_p(2,\F_d)$ be of type 2 and let the integer $r$ be as in the statement of Lemma \ref{lem2}.
\begin{enumerate}[\hskip 4mm 1.]
\item When $n$ is even, $G$ has suborder of at most $(d+1)/2$.
\item When $n$ is odd, $G$ has suborder $d+1$ if and only if $\gcd(r,d+1)=1$.
\end{enumerate}
\end{theorem}

\begin{proof}
Let $\lambda_{\pm}$ be the eigenvalues of $G$ and the sequence $s_m$ be as defined in Lemma \ref{lem1}. We recall that the suborder of $G$ is the smallest positive integer $m$ for which $s_m= 0$, which is equivalent to $\lambda_{+}^m = \lambda_{-}^m$ in this case when the two eigenvalues are distinct because $G$ is of type 2.
\begin{enumerate}
\item When $n$ is even, $(d-1)/(p-1) = 1+p+\dotsc+p^{n-1}$ is an even integer, so $(d-1)/2$ is a multiple of $(p-1)$. It then follows from Eq. (\ref{eqthetak}) that
\begin{equation}
\bar{\theta}^{k(d-1)(d+1)/2}=1,
\end{equation}
which implies $\lambda_{+}^{(d+1)/2} = \lambda_{-}^{(d+1)/2}$, or $s_{(d+1)/2}=0$. Therefore $G$ has suborder of at most $(d+1)/2$ and cannot be a MUB-cycler.

\item When $n$ is odd, $(d-1)/(p-1)$ is an odd integer. It follows from this, and the fact that $\gcd(d+1,d-1)=2$, that $(d-1)/(p-1)$ is co-prime to $d+1$. We have $\lambda_{+}^m = \lambda_{-}^m$ if and only if $\eta^{m(d-1)r}=1$, which in turn is true if and only if $mr(d-1)/(p-1)$ is a multiple of $d+1$. Therefore $G$ has suborder $d+1$ if and only if $r$ is co-prime to $d+1$.
\end{enumerate}
\end{proof}
In summary, in this section we have proved the nonexistence of MUB-cyclers when the exponent $n$ is even. In the case $n$ is odd, we have identified all elements in $\GL(2,\F_d)$ that give rise to MUB-cyclers according to the characteristics of their eigenvalues. Lastly, we want to provide an explicit form for these MUB-cycling elements. The proof in the Appendix of ref.  \cite{Subhash} can be easily extended to show that for any element in $G \in \GL(2,\F_d)$ with trace $t$ and determinant $\Delta$, where $t^2-4\Delta \ne 0$, there exists $S \in \SL(2,\F_d)$ such that
\begin{equation}\label{defcanonical}
G = S G_c S^{-1} \hskip 15mm
G_c = \begin{pmatrix}
0 & -\Delta \\
1 & t
\end{pmatrix} \ . \end{equation}
Therefore, an element of $\GL(2,\F_d)$ is a MUB-cycling matrix if and only if it is conjugate to $G_0^r$ where
\begin{equation}
G_0 = \begin{pmatrix}
0 & -\eta^{(d+1)}\\
1 & \eta + \eta^{d}
\end{pmatrix} \ ,
\label{eq:G0Def}
\end{equation}
$\eta$ is defined as in Eq. (\ref{defeta}), and  $r$ is an integer co-prime to $d+1$.  Note that the order of $G_0$ is $(p-1)(d+1)$ (because this is the smallest integer $r$ such that $\eta^r$ and $\eta^{dr}$,  the eigenvalues of $G^r_0$ are both equal to $1$).

It follows from this that anti-symplectic MUB-cycling matrices exist if and only if $d=3\mod 4$ (a fact already shown in ref.~\cite{Appleby}).  In fact $G_0^r$ is anti-symplectic if and only if $\eta^{r(d+1)}=-1$, which in turn is true if and only if $r$ is an odd multiple of $(p-1)/2$.  If $d=1\mod 4$ then $(p-1)/2$ is even and so no multiple of $(p-1)/2$ of is co-prime to $d+1$.  If, on the other hand, $d=3\mod 4$ it is easily seen that $(p-1)/2$ is co-prime to $d+1$ implying that $G_0^{r(p-1)/2}$ is a MUB-cycling anti-symplectic for every  $r$ co-prime to $d+1$.
\section{Eigenvectors of MUB-cyclers}
We now come to the question of finding the eigenvectors of a $g$-unitary.  The result we prove will play a crucial role in our construction in the next section, of MUB-balanced states in prime power dimensions equal $3 \mod 4$.

An ordinary unitary is, of course, always diagonalizable.  However the situation with $g$-unitaries is  more vexed.  Indeed, it is not guaranteed that an arbitrary $g$-unitary will have any eigenvectors at all.  This is easily seen in the special case of an anti-unitary.  Suppose that $U$ is an anti-unitary and $|\psi\ra$ an eigenvector, so that
\be
U |\psi\ra = \lambda |\psi\ra
\ee
for some $\lambda$.  It follows from the definition of the adjoint of an anti-unitary \cite{Wigner} that
\ea{
|\lambda|^2\la \psi, \psi \ra = \la U \psi , U\psi\ra = \left(\la \psi, U^{\dagger}U \psi\ra\right)^{*} = \la \psi, \psi \ra
} 
(where we have temporarily switched from Dirac notation to the notation usual in pure mathematics).  So $\lambda$ is a phase.  Consequently
\be
U^2 |\psi\ra = |\lambda|^2 |\psi\ra = |\psi\ra \ .
\label{eq:antiUnitarySquare}
\ee
We conclude that $|\psi\ra$ must be an eigenvector of the unitary $U^2$ with eigenvalue $1$.  It follows that if $U^2$ does not have any eigenvectors with eigenvalue $1$, then $U$ does not have any eigenvectors at all.  For an example consider the anti-unitary $U$ which acts on $\fd{C}^2$ according to
\be
U \bmt x \\ y \emt = \frac{1}{\sqrt{2}} \bmt y^{*}+x^{*} \\ y^{*} - x^{*}\emt \ .
\ee 
Since the eigenvalues of $U^2$ are $\pm i$, $U$ has no eigenvectors.

As we will see analogous statements hold in the case of an arbitrary $g$-unitary (except that in the case of a $g$-unitary which is not an anti-unitary the eigenvalues, if they exist, do not have to be phases).

Another important difference between $g$-unitaries and ordinary unitaries is that multiplying an eigenvector by a scalar can change the eigenvalue. Again, this is most easily seen in the special case of an anti-unitary.  Thus, if  $U$ is an anti-unitary and $|\psi\ra$ an eigenvector with eigenvalue $e^{i\theta}$, then $e^{i\phi}|\psi\ra$ is an eigenvector with eigenvalue $e^{i(\theta-2\phi)}$.  This means, in particular, that in the case of an anti-unitary, if one adjusts  the overall phase appropriately, one can always ensure that the eigenvalue is $1$. 

Wigner \cite{Wigner} analyzed the eigenvectors of anti-unitaries.  The problem of extending his analysis to the case of an arbitrary $g$-unitary is not straightforward.  In this section  we will confine ourselves to  $g$-unitaries of a very special kind:  namely the MUB-cycling $g$-unitaries defined in the last section.  For such $g$-unitaries   it is not difficult to give a complete characterization of their eigenvectors and eigenvalues.  The results of our analysis are summarized in theorem~\ref{thm:gUnitEigenvectors}.   The theorem states that   $g$-unitaries of the kind we consider always do have eigenvectors.  Moreover, their eigenvectors are confined to a one dimensional subspace (in other words, they are unique, up to multiplication by a scalar).  Finally, as with anti-unitaries, one can always adjust the overall scale factor so as to ensure that the eigenvalue is $1$.  

\begin{theorem}
\label{thm:gUnitEigenvectors}
Let $d = p^n$ where $n$ is odd, and let $G\in\GL_p(2,\fd{F}_d)$ be  a MUB-cycler.  Let $V_{\rm c}$ be the subspace of the full Hilbert space consisting of all vectors whose standard basis components are in the cyclotomic field $\fd{Q}(\omega)$. Then
\begin{enumerate}
\item There exists a non-zero vector $|\psi\ra\in V_{\rm c}$ such that 
\be
U_G |\psi\ra = |\psi\ra \ .
\ee
\item Let $|\phi\ra \in V_{\rm c}$ be arbitrary.  Then $|\phi\ra$ is an eigenvector of $U_G$ if and only if $|\phi\ra = \mu |\psi\ra$ for some $\mu \in \fd{Q}(\omega)$.
\item Let $2m_0$ be the smallest positive integer (necessarily even) such that $U_G^{2m_0}$ is unitary. Then the eigenspace of $U_G^{2m_0}$ with eigenvalue $1$ is one-dimensional and is spanned by $|\psi\ra$.
\item $|\psi\ra$ is an eigenvector of the parity operator, with eigenvalue $(-1)^{\frac{p-1}{2}}$.
\end{enumerate}
\end{theorem}
\begin{remark}
The vector $|\psi\ra$ cannot be assumed to be normalized (since it may happen that the normalization constant is not in the field $\fd{Q}(\omega)$, and since, even if it is in the field, the normalized vector may not have eigenvalue $1$).

The fact that the eigenvectors of $U_G$ are also eigenvectors of $U^{2m_0}_G$ with eigenvalue 1 is important as it provides us with a means of calculating them.

As a special case of this theorem, the MUB cycling anti-unitaries, whose existence in dimension $d=3\mod 4$ was established in ref. \cite{Appleby}, all have eigenvectors which are unique up to scalar multiplication.

In the fourth statement of the theorem, the term ``parity operator'' refers to the unitary
\be
A = \sum_{x} |-x\ra \la x |
\ee
which already appeared in eq. (\ref{ppA}). It follows from eq. (\ref{1}) and Lemma 1 in ref. \cite{Appleby} that
\ea{
A &= (-1)^{\frac{p-1}{2}} U_P
& &\text{where} &
P &= \bmt -1 & 0 \\ 0 & -1\emt \ .
\label{eq:ParityTermsP}
}
\end{remark}
\begin{proof}
Immediate consequence of Lemmas \ref{lm:UG2m0Criterion}--\ref{lm:parity} proved below.
\end{proof}
In the remainder of this section it will always be assumed that the exponent 
$n$ is odd, and that $G$ is a fixed element of $\GL_p(2,\fd{F}_d)$ with eigenvalues $\eta^r$ and $\eta^{rd}$ (as in Lemma \ref{lem2} and Theorem \ref{thm2})
where $r$ is co-prime to $d+1$ so that $G$ is a MUB-cycler. We will always write $t=\Tr(G)$ and $\Delta = \det(G)$.
From Theorem \ref{thm1} we know
\be
G^{d+1}=\Delta I \ .
\label{eq:GpowerdPlus1}
\ee
So, if we define the multiplicative order of $\Delta$ to be the smallest positive integer $m$ such that $\Delta^m=1$, this is the same as the smallest positive integer $m$ such that $U^m_G$ is a unitary. We have that $\Delta^m =1$ if and only if $mr(d^2-1)/(p-1)$ is a multiple of $d^2-1$, which, in turn, is true if and only if $mr$ is  a multiple of $p-1$.  Since $r$  is odd this means that the multiplicative order of $\Delta$ must be even.  We will therefore denote it $2m_0$.
\begin{lemma}
\label{lm:UG2m0Criterion}
With notations and definitions as above, suppose $|\phi\ra\in V_{\rm c}$ is an eigenvector of $U_G$ so that
\be
U_G |\phi\ra = \lambda |\phi\ra
\label{eq:uglambda}
\ee
for some $\lambda\in \fd{Q}(\omega)$.  Then
\be
U^{2m_0}_G |\phi\ra = |\phi \ra \ .
\ee
\end{lemma}
\begin{proof}
The fact that $\Delta^{2m_0}= 1$ implies $\Delta^{m_0} = \pm 1$.  Since $2m_0$ is the smallest positive integer $m$ such that $\Delta^m = 1$ we must in fact have $\Delta^{m_0}=-1$.  This means that $U^{m_0}_G$ is an anti-unitary.  Repeatedly applying $U_G$ to Eq.~(\ref{eq:uglambda}) gives
\be
U_G^{m_0} |\psi\ra = \kappa |\psi\ra
\ee
where $\kappa = \lambda g^{\vphantom{m_0}}_{\Delta} (\lambda) \dots g^{m_0-1}_{\Delta}(\lambda)$.  It then follows from the discussion at the beginning of this section [see Eq.~(\ref{eq:antiUnitarySquare})] that 
\be
U^{2m_0}_G |\psi\ra = |\psi\ra \ .
\ee
\end{proof}
\begin{lemma}
\label{lm:SjDim}
With notations and definitions as above, let $\mathcal{S}_{1}$ be the eigenspace of $U^{2m_0}_G$ corresponding to the eigenvalue $1$.  Then $\dim(\mathcal{S}_1)=1$.
\end{lemma}
\begin{proof}
Taking account of the fact that the representation of $\SL(2,\fd{F}_d)$ is faithful, the projector onto $\mathcal{S}_1$ is 
\be
P_1 = \frac{1}{d+1} \sum_{u=0}^{d}U_{G^{2m_0u}} \ ,
\label{eq:musProj}
\ee
implying
\be
\dim S_{1} = \frac{1}{d+1}\sum_{u=0}^{d}  \Tr\bigl(U_{G^{2m_0u}}\bigr) \ .
\ee

We will use a slightly improved version of Theorem 5 of ref. \cite{Appleby} to evaluate this expression. The original theorem states that given $S \in \SL(2,\fd{F}_d)$ of the form 
\be S = \bmt \alpha & \beta \\ \gamma & \delta \emt \ee 
with determinant 1 and trace $t \ne 2$, then 
\be \Tr (U_S) = 
\begin{cases}
l(t-2) \qquad &\beta \ne 0 \\
l(\alpha) 	\qquad &\beta = 0 \ .
\end{cases}
\ee
We improve this statement by showing that for $t \ne 2$, $\Tr (U_S) = l(t-2)$ irrespective of whether $\beta = 0$. Indeed, consider the case $\beta = 0$, when $S$ takes the form
\be S = \bmt \alpha & 0 \\ \gamma & -\alpha^{-1} \emt \ .\ee 
We have $t = \alpha + \alpha^{-1}$ and $l(t-2) = l\left((\alpha^{1/2} - \alpha^{-1/2})^2\right)$. Note that
\be (\alpha^{1/2} - \alpha^{-1/2})^d = \alpha^{d/2} - \alpha^{-d/2} = 
\begin{cases}
\alpha^{1/2} - \alpha^{-1/2} \quad & \alpha^{1/2} \in \fd{F}_d \\
\alpha^{-1/2} - \alpha^{1/2} \quad & \alpha^{1/2} \notin \fd{F}_d \ ,
\end{cases}
\ee
which means $\alpha^{1/2} - \alpha^{-1/2} \in \fd{F}_d \ $  if and only if $\alpha^{1/2} \in \fd{F}_d$, and consequently $l(\alpha) = l(t-2)$, as we wish to show.

Back to evaluating $\Tr(G^{2m_0u})$, since $G$ has eigenvalues $\eta^r$ and $\eta^{dr}$, we have
\ea{
\Tr(G^{2m_0 u}) &= \eta^{2rm_0u} +\eta^{2drm_0u}
\nonumber
\\
& = \eta^{2rm_0u} + \eta^{-2rm_0u}
\nonumber
\\
&= \left(\eta^{rm_0u}-\eta^{-rm_0u}\right)^2 +2
}
where in the penultimate step we used the fact that $\eta^{2drm_0} = \eta^{-2rm_0}$ (as follows from the fact that $\Delta^{2m_0} = 1$).   We see from this that $\Tr(G^{2m_0u})=2$ if and only if $2rm_0u$ is a multiple of $(p-1)(d+1)$.  Since $2m_0$ is the order of an element in a group of order $p-1$ it must be a divisor of $p-1$.  So the condition is equivalent to the statement that $ru$ is a multiple of $(d+1)(p-1)/2m_0$, which, in view of the fact that $r$ is co-prime to $d+1$, is in turn equivalent to the statement that $u$ is a multiple of $(d+1)$.  Since $u$ is in the range $0\le u \le d$ this means $\Tr(G^{m_0u})=2$ if and only if $u=0$. Applying the improved result from Theorem 5 of ref. \cite{Appleby} we find
\be
  \Tr\bigl(U_{G^{2m_0u}}\bigr)
 =
 \begin{cases}
 d \qquad & u=0 \\
 l\left((\eta^{rm_0u}-\eta^{-rm_0u})^2 \right)\qquad & 1 \le u \le d
 \end{cases}
 \label{eq:trCalcInter}
\ee
To evaluate the quadratic characters we again appeal to the fact that $\Delta^{m_0} = -1$, which implies $\eta^{drm_0} = -\eta^{-rm_0}$, so that
\ea{
\left(\eta^{rm_0u}-\eta^{-rm_0u}\right)^d
&= \eta^{drm_0u} - \eta^{-drm_0u}
\nonumber
\\
&=(-1)^{u+1}(\eta^{rm_0u}-\eta^{-rm_0u}) \ .
}
So for $1 \le u \le d$, we have
\be
 l\left((\eta^{rm_0u}-\eta^{-rm_0u})^2 \right) = (-1)^{u+1} 
\ee
and, consequently,
\ea{
\dim \mathcal{S}_1 &= \frac{1}{d+1}\left(d + \sum_{u=1}^d (-1)^{u+1}\right)
= 1 \ .
}
\end{proof}

\begin{lemma}
\label{lm:psiExist}
With definitions and notations as above there exists a non-zero vector $|\psi\ra \in V_{\rm c}$ such that
\be
U_G |\psi\ra = |\psi\ra \ .
\ee
\end{lemma}
\begin{proof}
We know from Lemma \ref{lm:SjDim} that $\det(U_{G^{2m_0}}-I) = 0$.  Since the matrix elements of $U_{G^{2m_0}}$ are in the field $\fd{Q}(\omega)$ we have, by a basic fact of linear algebra~\cite{Greub}, that there exists a vector $|\phi\ra \in V_{\rm c}$ such that $U_{G^{2m_0}}|\phi\ra = |\phi\ra$.  Since 
\be
U_{G^{2m_0}} U_{G} |\phi \ra = U_G |\phi\ra
\ee
and since $\mathcal{S}_1$ is one dimensional, we must have
\be
U_G |\phi\ra = \lambda |\phi\ra
\ee
for some $\lambda \in\fd{Q}(\omega)$ such that
\be
\lambda g_{\Delta} (\lambda) \dots g^{2m_0-1}_{\Delta} (\lambda)= 1.
\ee
It now follows from a variant  of the proof of Hilbert's theorem 90 (see, for example, refs. \cite{Milne} or~\cite{Roman} [we cannot  make a direct application of the theorem since $g_{\Delta}$ may not generate the full Galois group]) that there exists $\mu \in \fd{Q}(\omega)$ such that
\be
\lambda = \frac{\mu}{g_{\Delta}(\mu)} \ .
\label{eq:lammugmu}
\ee
In fact, define a mapping $T\colon \fd{Q}(\omega) \to \fd{Q}(\omega)$ by 
\ea{
T(x) &= \sum_{j=0}^{2m_0-1} a_j  g^j_{\Delta}(x)  
}
where
\ea{
a_0 &= 1, & a_1 &= \lambda, & a_2 &= \lambda g_{\Delta}(\lambda), & &\dots, & a_{2m_0-1} &= \lambda g^{\vphantom{m_0}}_{\Delta}(\lambda) \dots g^{2m_0-2}_{\Delta}(\lambda) \ .
}
We have
\be
\lambda g_{\Delta}(T(x)) = T(x)
\ee
for all $x$.  
Since the coefficients $a_j$ are non-zero it follows from the Dedekind independence theorem (see, for example, refs.~\cite{Milne} or~\cite{Roman}) that there exists $x\in \fd{Q}(\omega)$ such that $T(x) \neq 0$. Choose such a value of $x$ and set $\mu = T(x)$.  Eq.~(\ref{eq:lammugmu}) then follows. If we now define $|\psi\ra = \mu |\phi\ra$ we will have
\be
U_G |\psi\ra = |\psi\ra \ .
\ee
\end{proof}
\begin{lemma}
\label{lm:parity}
With definitions and notations as above, let $|\phi\ra$ be such that $U_G |\phi\ra = \lambda |\phi\ra$ for some $\lambda$.  Then $|\phi\ra$ is an eigenvector of the parity operator, with eigenvalue $(-1)^{\frac{p-1}{2}}$.
\end{lemma}
\begin{proof}
It follows from Lemmas \ref{lm:UG2m0Criterion} and \ref{lm:SjDim} that 
\be
U_{G}^{2m_0} |\phi\ra = |\phi\ra \ .
\ee
Since $\Delta^{2m_0} = 1$, and since $2m_0$ is the smallest positive integer for which that is the case, we must have $\Delta^{m_0}=-1$.  In view of Eqs. (\ref{eq:ParityTermsP}) and (\ref{eq:GpowerdPlus1}) this means
\ea{
A |\phi\ra& = (-1)^{\frac{p-1}{2}} \left( U_G^{d+1}\right)^{m_0} |\phi\ra
\nonumber
\\
& = (-1)^{\frac{p-1}{2}} \left(U_{G}^{2m_0}\right)^{\frac{d+1}{2}} |\phi\ra
\nonumber
\\
& =(-1)^{\frac{p-1}{2}} |\phi\ra \ .
} 
\end{proof}

We conclude this section by observing that since every cycling $g$-unitary has exactly one eigenvector up to scalar multiplication, then $|\psi\ra$ is an eigenvector of $U_{G^r_0}$ if and only if it is an eigenvector of $U_{G_0}$ (where $G_0$ is the matrix defined by Eq. (\ref{eq:G0Def})). In view of the discussion at the end of Section \ref{sec:MUBcyclers} this means that the vectors which are eigenvectors of a cycling $g$-unitary form a single orbit under the extended Clifford group.
\section{MUB-balanced states}
We saw in the last section that if $d$ is an odd power of $p$  then every MUB-cycling $g$-unitary has an eigenvector.  We now show that if, in addition, $d=3\mod 4$ then these eigenvectors are ``rotationally invariant'' or ``MUB-balanced'' states \cite{Sussman,Sussmanthesis,ASSW} and, \emph{a fortiori}, minimum uncertainty states \cite{Sussman,ADF,Sussmanthesis,ASSW}.  The question as to what happens when $d=1\mod 4$ remains open.

Given a MUB, a MUB-balanced state is one for which the probabilities with respect to each basis are permutations of each other.  In other words a normalized state $|\psi\ra$ is MUB-balanced if and only if for all $z$ there is a permutation $f_z$ such that
\be
p_j^{(z)} = p^0_{f_z(j)}
\ee
where
\be
p^{(z)}_j = \bigl| \la e^{(z)}_j | \psi \ra \bigr|^2 
\ee
This definition was introduced by Wootters and Sussman \cite{Sussman}, who also showed that such states exist in every even prime power dimension.  Wootters and Sussman went on to show that MUB-balanced states are minimum uncertainty states.  Since it is  central to this section it is worth summarizing their argument.  Let
\be
H_z = -\log_2 \left(\sum_j (p^{(z)}_j)^2\right)
\ee
be the quadratic R\'{e}nyi entropy in basis $z$, and let $T=\sum_z H_z$ be the total entropy.  Then can be shown that $T$ satisfies the inequality
\be
T \ge (d+1)\log_2\bigl(\frac{d+1}{2}\bigr)
\ee
A minimum uncertainty state is one for which the bound is saturated.  The necessary and sufficient condition for that to be true is that
\be
\sum_j (p^{(z)}_j)^2 = \frac{2}{d+1}
\label{eq:MUScond}
\ee
In a MUB-balanced state the fact that probabilities in each basis are the same up to permutation in every basis means that the sum  $\sum_j (p^z_j)^2$ is independent of $z$.  In view of the identity
\be
\sum_{j,z} (p^{(z)}_j)^2 = 2
\ee
it follows that Eq. (\ref{eq:MUScond}) is satisfied, and that the state is consequently a minimum uncertainty state.

Our result is described by the following theorem
\begin{theorem}
Suppose $d=3 \mod 4$, and suppose $G\in\GL_p(2,\fd{F})$ is a MUB-cycler.  Let $2m_0$ be the integer defined in Theorem 1 of Section 9, and let $|\phi\ra$ be a normalized eigenvector of the unitary $U^{2m_0}_G$ with eigenvalue 1.  Then $|\phi\ra$ is  MUB-balanced.
\end{theorem}
\begin{remark}
Note that, although the $g$-unitary $U_G$ plays an essential role in the proof, a knowledge of the ordinary unitary $U_{G^{2m_0}}$ is sufficient if one only wants to calculate the state. 

Since the eigenstates of MUB-cycling $g$-unitaries form a single orbit of the extended Clifford group we could restrict our attention to the the eigenstates of MUB-cycling anti-unitaries.  This would make the proof a little easier.  Nevertheless, we have chosen to give the proof for the case of an arbitrary $g$-unitary because this enables one to see  why it fails when $d=1\mod 4$.
\end{remark}
\begin{proof}
We know from Theorem 1 of Section 9 that there exists a state  $|\psi\ra$ such that
\be
U_G |\psi\ra = |\psi \ra
\ee
 Since $U_G$ is a cycling $g$-unitary the set $\{z_0,z_1,\dots,z_d\}$ is a permutation of the full set of labels $\{0,1,\dots, d-1,\infty\}$.  We may write
\be
|e^{z_l}_j\ra = s_{l,j} \omega^{k_{l,j}} U_G^l |e^{0}_{p_l(j)} \ra
\ee
where $p_l$ is an $l$-dependent permutation, $k_{l,j}$ is an $l$ and $j$-dependent integer and $s_{l,j}$ is an $l$ and $j$-dependent sign.  So
\ea{
\la e^{z_l}_{p_l(j)} | \psi \ra
&=
s_{l,j} \omega^{-k_{l,j}} g_{\Delta}^{-l} \left(\la e^{0}_j |U^{\dagger}_j| \psi\ra\right)
\nonumber
\\
&= s_{l,j} \omega^{-k_{l,j}} g_{\Delta}^{-l} \left(\la e^{0}_j | \psi\ra\right)
}

Now let
\be
F=\bmt \alpha & \beta \\ \gamma & \delta \emt
\ee
be arbitrary.  We have
\ea{
K_{\Delta} F K_{\Delta}^{-1}
&= \bmt \alpha & \Delta^{-1} \beta \\ \Delta \gamma & \delta \emt
}
Since $\Delta =\theta^{\frac{r(d-1)}{(p-1)}}$, and since
  $r(d-1)/(p-1)$  is odd,  $\Delta$ is a quadratic non-residue.  The assumption that $d=3 \mod 4$, together with Lemma~1 of ref.~\cite{Appleby}, means that $-1$ is also a quadratic non-residue.  So  there exists $x\in \fd{F}_d$ such that $x^2 = -\Delta$ .  If we define
\be
S =\bmt x^{-1} & 0 \\ 0 & x \emt
\ee 
then
\ea{
S K_{-1} F K_{-1} S^{-1}& = K_{\Delta} F K^{-1}_{\Delta}
}
So
\ea{
g_{\Delta} (U_F) 
&= U_S U^{*}_F U^{-1}_S
}
Applying this to the projector $P_1$ defined in Eq.~(\ref{eq:musProj}) we deduce
\be
g_{\Delta} (P_1) = 
U_S P^{*}_1 U_S^{-1} 
\ee
We have
\be
\la e^{0}_{j_1} | \psi \ra \la \psi | e^{0}_{j_2}\ra
= 
\lambda \la e^{0}_{j_1} | P_1 | e^{0}_{j_2}\ra 
\ee
for some constant $\lambda$.  Consequently
\ea{
g_{\Delta}(\la e^{0}_{j_1} | \psi\ra) g_{\Delta} (\la \psi | e^{0}_{j_2} \ra)
&=\frac{g_{\Delta}(\lambda)}{\lambda^{*}}(\la e^{0}_{x^{-1}j_1} | \psi\ra)^{*} (\la \psi | e^{0}_{x^{-1}j_2} \ra)^{*}
}
which is easily seen to imply
\be
g_{\Delta} (\la e^{0}_{j}|\psi\ra) = \mu (\la e^{0}_{x^{-1}j} | \psi\ra)^{*}
\ee
for some constant $\mu$.  By repeated application of this formula we find
\be
g^{l}_{\Delta}  (\la e^{0}_{j}|\psi\ra)
=
\begin{cases}
\mu_l (\la e^{0}_{x^{-l}j} | \psi\ra)^{*}  \qquad & \text{$l$ odd}
\\
\mu_l \la e^{0}_{x^{-l}j} | \psi\ra \qquad & \text{$l$ even}
\end{cases}
\ee
where $\mu_l = g_{\Delta}^{l-1} (\mu) g_{\Delta}^{l-2}(\mu^{*}) \dots$.  Hence
\be
\bigl| \la e^{z_l}_{p_l(j)} |\psi\ra \bigr|^2 = |\mu_l|^2 \bigl| \la e^{0}_{x^{-l}j} |\psi \ra \bigr|^2
\ee
Since
\be
\sum_j \bigl| \la e^{z_l}_{j} |\psi\ra \bigr|^2  = \sum_j \bigl| \la e^{0}_{j} |\psi \ra \bigr|^2 =1
\ee
we must have $|\mu|^2=1$, from which it follows that the normalized state
\be
|\phi\ra = \frac{1}{\sqrt{\bigl\| |\psi\ra \bigr\|}} |\psi\ra
\ee
is MUB-balanced.
\end{proof}

The MUB-balanced states whose existence is established by this theorem are identical with the ones constructed using a different method by Amburg \emph{et al}~\cite{ASSW}.  In fact, the orbit of states constructed in ref.~\cite{ASSW} is generated by the state with Wigner function
\be
W_{\mbf{p}} = \frac{1}{d(d+1)} \left(1-d \delta_{\mbf{p},\boldsymbol{0}} + \sum_{x\in \fd{F}_d^{*}}l(x^2+1) \omega^{\tr[x(p_1^2+p_2^2)]}\right).
\ee 
Define
\begin{align}
F &= \bmt \alpha & \beta \\ -\beta & \alpha \emt, & \alpha &= \frac{1}{2}(\eta+\eta^d) & \beta&=\frac{i_M}{2}(\eta-\eta^d)
\end{align}
where
\be
i_M = \eta^{\frac{1}{4}(p-1)(d+1)}
\ee
(observe that $i_M^2 = -1$, so $i_M$ is a modular analogue of $i$).  It is easily seen that $\alpha^d = \alpha$ and $\beta^d=\beta$.  So $F\in \GL(2,\fd{F}_d)$. Moreover
\begin{align}
\Delta &= \det(F) = \eta^{d+1}, &  t&= \Tr(F) = \eta+\eta^d.
\end{align}
So $F$ is a MUB-cycler.  The fact that  
\be
(\alpha p_1+\beta p_2)^2 + (- \beta p_1+ \alpha p_2 )^2 =  \Delta (p_1^2+p_2^2)
\ee
means $W_{F\mbf{p}} = g_{\Delta} (W_{\mbf{p}})$, which is easily seen to imply that the state corresponding to $W_{\mbf{p}}$ is an eigenstate of $U_F$ (\emph{modulo} multiplication by a scalar).  

It is interesting to ask how many MUB-balanced states there are.  Consider the MUB-cycler $U_{G_0}$, where $G_0$ is the matrix defined by Eq.~(\ref{eq:G0Def}) with multiplicative order $2 m_0 = p-1$.  Let $|\psi\ra$ be the corresponding MUB-balanced state.  Since $|\psi\ra$ is the unique eigenstate of $U^{p-1}_{G_0}$ with eigenvalue 1 it will be left invariant by any Clifford unitary or anti-unitary $V$ such that
\be
V U^{p-1}_{G_0} V^{-1} = U^{t(p-1)}_{G_0}
\ee 
for some $t$.  One finds that the only possibilities are 
\ea{
V = U_{G_0^{u(p-1)/2}}
}
(corresponding to $t=1$) or
\ea{
V&= U_{G_0^{u(p-1)/2} F} \ , & F &= \bmt 0 & i_M \eta^{\frac{d+1}{2}} \\
i_M \eta^{-\frac{d+1}{2}}  & 0 \emt
}
(corresponding to $t=-1$), where in both cases $u$ can take any integral value in the range $0 \le u < 2(d+1)$ (to understand the form of the matrix $F$ observe that $FG_0F^{-1} = \eta^{d+1} G_0^{-1}$, implying $FG^{p-1}_0 F^{-1} = G^{-(p-1)}_0$).  We conjecture that there are no other extended Clifford unitaries or anti-unitaries leaving $|\psi\ra$ invariant (we have checked this conjecture in detail for $d=7, 11, 19$).  Since the order of the extended Clifford group is $2d^3(d^2-1)$ (see ref.~\cite{Gehles}) that  would mean that the number of MUB-balanced states is $d^3(d-1)/2$.

We were surprised by the results obtained in this section, as we had expected that our construction would yield many new MUB-balanced states, additional to the ones described by Amburg \emph{et al}.     However, our results suggest  that Amburg \emph{et al} have in fact constructed the entire set.  If so it would mean that such states form a highly distinguished geometrical structure.  Amburg \emph{et al} describe MUB-balanced states as ``states which have no right to exist''.  The same could be said of SICs.  But whereas there are, in most of the dimensions which have been examined, several different orbits of SICs, it looks as though there is only one orbit of MUB-balanced states.  If that were indeed the case it would mean that such states are very special indeed.
\section{Summary}
We have had to go through a large amount of detailed arguments, and our main results 
are summarized in theorems 1, 2, 5, 6, and 12. At the same time the 
picture we have arrived at is simple and appealing.

First, the g-unitaries themselves. They generalize the notion of anti-unitaries, but 
their action is restricted to vectors taking values in some special number field. 
We have been concerned with what is arguably the simplest example, when the number 
field is generated by some root of unity. Then the g-unitaries do play 
a role in the description of complete sets of MUB in odd dimensions, and 
as far as the transformations of the actual MUB vectors themselves are concerned 
they have a simple interpretation as rotations in Bloch space---just 
as ordinary unitaries always have. 

Mutually unbiased bases are interesting from many points of view. Here we have been 
interested in the projective transformations that permute them, and we saw 
that provided the dimension is a prime power $d = p^n$ where $n$ is odd the 
g-unitaries provide us with transformations that cycle through the entire set 
of $d+1$ bases. If the dimension equals 3 mod 4 some of these transformations 
are effected by anti-unitaries, but when the dimension equals 1 mod 4 it is 
necessary to turn to g-unitaries for this purpose. We have shown that every 
MUB-cycling g-unitary leaves one vector in Hilbert space invariant, and that 
this eigenvector has definite parity. If  $d = 3$ mod 4 this 
eigenvector is also an eigenvector of an anti-unitary operator which is itself 
MUB-cycling. (In even prime power dimensions unitary MUB-cyclers exist \cite{Sussman, 
Seyfarth}.)

If $d=3$ mod 4 the eigenvectors of MUB-cycling g-unitaries are MUB-balanced states in the sense of Wootters and Sussman~\cite{Sussman,Sussmanthesis} and Amburg et al.~\cite{ASSW}.  Our construction is a useful supplement to the work of Amburg et al. for two reasons.  In the first place it gives additional insight into the features of Hilbert space which are responsible for the existence of such states.  In their original paper Wootters and Sussman showed that one gets MUB-balanced states in the even prime power case by taking eigenvectors of cycling unitaries.  Amburg \emph{et al} demonstrated the existence of such states in odd prime power dimension equal to $3$ mod 4.  They also gave a very appealing formula for the Wigner function of such a state.  However, they did not demonstrate the connection with the even prime power case.  In this paper we have demonstrated such a connection:  to go from the even prime power case to the odd prime power case with $d=3$ mod 4 one merely has to replace the cycling unitaries of Wootters and Sussman with cycling $g$-unitaries.  In the second place the fact that we have exposed the underlying structural reasons for the existence of MUB-balanced states suggests the conjecture that we have in fact found \emph{every} such state.  We had expected that our construction would yield many new MUB-balanced states, additional to those found by Amburg et al.  However, our results suggest that such states are confined to a single orbit of the extended Clifford group.  If that is the case it would mean that they are very remarkable states indeed, which we feel may well repay further investigation.  In this connection we would particularly like to draw the reader's attention to the fact that  the distribution of their components 
is governed by Wigner's semi-circle law when $d$ is large \cite{ASSW} (a fact which has attracted the interest of the pure mathematics community~\cite{Katz}). 

Finally, as we noted in the introduction, the concept of a $g$-unitary originally arose in connection with the SIC-problem~\cite{AYAZ}.  Specifically the known SIC fiducials are all eigenvectors of a family of $g$-unitaries.   In this paper we have proved a number of results connected with the problem of finding the eigenvectors of a $g$-unitary.  Our discussion of this problem was not exhaustive, being confined to the case of MUB-cycling $g$-unitaries defined over a cyclotomic field.  The problem of extending our analysis to  the  $g$-unitaries which arise in the SIC problem is  non-trivial.  Nevertheless, our results may be regarded as a useful first step in the direction of solving that more difficult problem.

\section*{Acknowledgements}
IB thanks Jan-\AA ke Larsson for criticism at an early stage. DMA thanks Steve Flammia for  illuminating discussions concerning the eigenvectors of $g$-unitaries.  David Andersson 
provided useful information about MUB-balanced states. HBD was supported by the Natural Sciences and Engineering Research Council of Canada and the Vanier Canada Graduate Scholarship. DMA was supported by the IARPA MQCO program, by the ARC via EQuS project number CE11001013, and by the US Army Research Office grant numbers W911NF-14-1-0098 and W911NF-14-1-0103.

\end{document}